\documentclass{article}
\usepackage[
    a4paper,
    margin=1in,
]{geometry}

\usepackage{amsmath}
\usepackage{amsthm}
\usepackage{amssymb}

\usepackage{graphicx}
\usepackage{float}
\usepackage{booktabs}
\usepackage{array}
\usepackage{tabularx}
\usepackage{subcaption}
\usepackage[table]{xcolor}
\usepackage{tikz}
\usepackage{pgfplots}
\pgfplotsset{compat=1.18}
\usepackage{pgfplotstable}
\usepackage{xcolor}
\usetikzlibrary{arrows.meta,positioning}
\usepackage{pgfplots}
\pgfplotsset{compat=1.18}
\usepgfplotslibrary{fillbetween}
\usetikzlibrary{intersections}
\usepackage{xcolor}
 
\definecolor{catblue}{RGB}{55,138,221}
\definecolor{labelgray}{RGB}{110,110,110}
\usepackage{algorithm}
\usepackage{algpseudocode}

\usepackage{enumitem}

\usepackage{hyperref}
\hypersetup{
    colorlinks=true,
    linkcolor=blue,
    citecolor=blue,
    urlcolor=blue
}

\usepackage{siunitx}
\usepackage{xcolor}

\newtheorem{theorem}{Theorem}[section]
\newtheorem{proposition}[theorem]{Proposition}
\newtheorem{corollary}[theorem]{Corollary}

\newtheorem{assumption}[theorem]{Assumption}

\theoremstyle{remark}
\newtheorem{remark}[theorem]{Remark}

\newcommand{\Pbase}{P_{\mathrm{base}}}
\newcommand{\Chat}{\widehat{C}}
\newcommand{\xstd}{\tilde{\mathbf{x}}}

\definecolor{insuredcol}{RGB}{240,149,123}
\definecolor{uninsuredcol}{RGB}{226,75,74}
\definecolor{catbondcol}{RGB}{55,138,221}
\definecolor{labelgray}{RGB}{100,100,100}
\title{CAT Bond Pricing with Kolmogorov--Arnold Networks}
\author{Sean Seow Cheng Hong\thanks{This work is part of an undergraduate research project (UROPS) in the Department of Mathematics, National University of Singapore, supervised by Assistant Professor Julian Sester.}}
\date{\today}

\begin{document}
\maketitle

\begin{abstract}
We study the approximation of CAT bond prices under a compound Poisson
loss model with lognormal severities using Kolmogorov--Arnold Networks
(KANs). Building on a baseline-plus-residual learning pipeline, we
train a KAN on the deviation from a closed-form lognormal baseline
price and extract an interpretable symbolic pricing formula. The
extracted formula achieves an average relative pricing error of
$0.483\%$ on a fully disjoint holdout sample of $90{,}000$ simulated
prices. In contrast to purely empirical modelling, we also analyse
structural properties of the true CAT bond pricing map, including
monotonicity with respect to the catastrophe arrival intensity
$\lambda$, the initial short rate $r_0$, and the trigger threshold
$D$. We derive sufficient conditions on KAN edge functions that
guarantee these monotonicities are preserved by the learned model, and
formulate a monotonicity-constrained training objective with a
convergence guarantee. Our results suggest that symbolic KAN
surrogates provide a practical compromise between accuracy,
computational speed, and interpretability for CAT bond valuation.
\end{abstract}

\tableofcontents
\newpage

\section{Introduction}
\label{sec:intro}

\subsection{Background and Motivation}
\label{subsec:background}

A central tension in applied machine learning is the trade-off between
accuracy and interpretability. Neural networks are powerful function
approximators, but the representations they build are opaque: once
training is complete, there is generally no way to extract a
human-readable formula from the model. Kolmogorov--Arnold Networks
(KANs), introduced by Liu et al.\ \cite{KANpaper}, take a different
approach. Rather than applying fixed nonlinear activations at each
node, a KAN places \emph{learnable univariate functions} on the edges
of the network. This makes it far easier to inspect what the network
has learned and, in favourable cases, to extract a compact closed-form
expression. The promise is not just accurate prediction, but a
transparent formula that a user can inspect, stress-test, and
understand.

This paper applies that idea to the pricing of catastrophe (CAT) bonds.
CAT bonds are fixed-income instruments that transfer insurance risk ---
from earthquakes, hurricanes, and floods --- directly to capital-market
investors. If a qualifying disaster occurs before maturity, some or all
of the principal is redirected to the bond's sponsor to cover losses. If no disaster occurs, investors receive full repayment together with
above-market coupons that compensate them for bearing the tail risk.
According to Artemis \cite{Artemis2025}, the outstanding CAT bond market
reached a record \$61.3 billion at the end of 2025, reflecting 24\% growth from
the prior year and average annual growth of around 8\% over the past
decade. As the frequency and severity of natural disasters continue to
rise, fast and accurate pricing of these instruments is increasingly
important.
The standard pricing framework, due to Baryshnikov et al.\
\cite{Baryshnikov1998} and refined by Burnecki and Kukla
\cite{BurneckiKukla}, models aggregate losses as a compound Poisson
process under a stochastic interest rate environment and computes bond
prices via Monte Carlo simulation. This approach is accurate, but slow, as re-evaluating prices across many parameter configurations required in risk management or sensitivity analysis is computationally expensive. Sester and Xu \cite{CATbondPaper} showed
that a feedforward neural network can serve as a fast, accurate
surrogate for this pricing map, achieving order-of-magnitude speedups.
The limitation is that such networks are black boxes: they produce
predictions without explanation.

Our goal is to show that the KAN can achieve competitive accuracy on CAT bond pricing while going one step further: producing a symbolic approximation of the pricing function that is compact, interpretable, and mathematically tractable.

\subsection{Related Literature}
\label{subsec:literature}

This work sits at the intersection of three bodies of literature: CAT
bond pricing, neural network models for financial
derivatives, and Kolmogorov--Arnold Networks.

\paragraph{CAT bond pricing.}
The mathematical pricing of CAT bonds under a compound Poisson
framework was pioneered by Baryshnikov et al.\ \cite{Baryshnikov1998},
who introduced a no-arbitrage approach using doubly stochastic Poisson
processes. Burnecki and Kukla \cite{BurneckiKukla} then derived
explicit pricing formulae for zero-coupon and coupon CAT bonds and
calibrated the model to real catastrophe loss data.
H\"{a}rdle and L\'{o}pez \cite{HardleLopez2010} applied this approach
to the calibration of CAT bonds for Mexican earthquake risk.
Later extensions include stochastic interest rate dynamics via the CIR
model \cite{NowakRomaniuk2018}, joint analysis of multiple risk drivers
such as loss severity, claim intensity, and trigger thresholds
\cite{MaMa2013}, and a risk-neutral framework with time-varying
catastrophe intensities \cite{Hofer2021}. A systematic survey of
compound Poisson approaches to CAT bond pricing is provided by
Sukono et al.\ \cite{Sukono2022}. A common limitation across this
literature is computational. This is because when the loss distribution does not yield
a closed-form survival probability, prices must be evaluated by Monte
Carlo simulation, which motivates the surrogate modelling approach
taken here.

\paragraph{Neural network surrogates for financial derivatives.}
Using neural networks as fast pricing surrogates is well established
in the option pricing literature. Early work by Garcia and Gen\c{c}ay
\cite{GarciaGencay2000} and Anders et al.\ \cite{Anders2001} showed
that feedforward networks can approximate option pricing functions at
near-instantaneous speed, with theoretical justification provided by
the universal approximation theorem \cite{Hornik1989}. More recently,
deep learning methods have been applied across a wider range of
derivatives and models, including neural SDE approaches
\cite{FanLiu2024} and finance-informed architectures that embed
no-arbitrage constraints directly into training \cite{FINN2024}. In
the CAT bond setting specifically, Sester and Xu \cite{CATbondPaper}
demonstrated that feedforward networks trained on Monte Carlo simulated
prices achieve high accuracy, and proved universal approximation results
for the CAT bond pricing map on compact domains. This paper builds
directly on their framework, replacing the feedforward network with a
KAN to obtain an interpretable symbolic formula.

\paragraph{Kolmogorov--Arnold Networks.}
KANs were introduced by Liu et al.\ \cite{KANpaper} as a
practically motivated reinterpretation of the Kolmogorov--Arnold
representation theorem. Since their introduction, they have been
applied in scientific machine learning settings where interpretability
is a priority, including physics-informed modelling and symbolic
regression tasks. However, their use in financial mathematics and, in
particular, derivative pricing is largely unexplored. This paper therefore attempts to apply KANs to CAT bond pricing.

\subsection{Kolmogorov--Arnold Networks}
\label{subsec:kans}

\subsubsection{From fixed activations to learnable edge functions}

A standard feedforward neural network applies a sequence of linear
maps interleaved with fixed nonlinear activations. For example, a
two-layer network with ReLU activations computes
\[
    f(\mathbf{x}) = W_2\,\sigma(W_1\mathbf{x} + b_1) + b_2,
\]
where $W_1, W_2$ are learned weight matrices, $b_1, b_2$ are learned
biases, and $\sigma$ is the elementwise ReLU. The network adapts to
the data by adjusting the weights and biases, but the shape of the
nonlinearity --- the shape of the ReLU --- is fixed throughout.

A KAN \cite{KANpaper} changes this design. Instead of placing fixed
nonlinearities on nodes, a KAN places \emph{learnable univariate
functions} on the edges, and sums at the nodes. A two-layer KAN with
input $\mathbf{x} = (x_1,\dots,x_d)$ computes
\[
    f(\mathbf{x})
    =
    \sum_{j=1}^{h}
    \phi_{2,j}\!\left(
        \sum_{i=1}^{d} \phi_{1,ij}(x_i)
    \right),
\]
where $\phi_{1,ij}:\mathbb{R}\to\mathbb{R}$ and
$\phi_{2,j}:\mathbb{R}\to\mathbb{R}$ are all learned during training.
In practice, each edge function is parameterised as a B-spline — a piecewise polynomial whose shape is learned during training. The grid size and spline order are hyperparameters chosen before training.

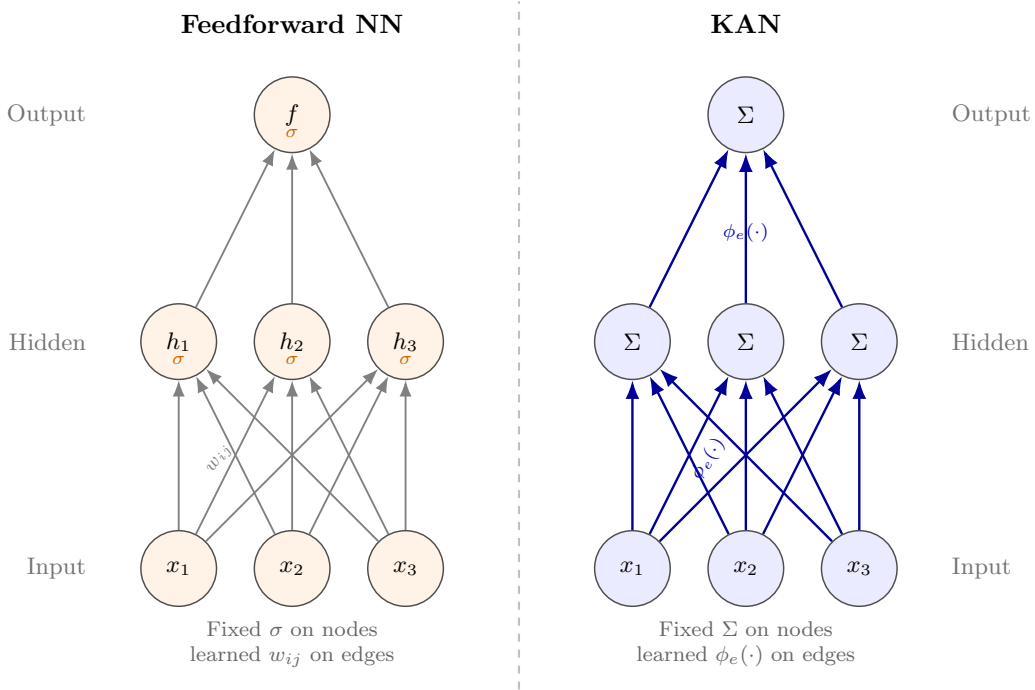
\begin{figure}[H]
\centering
\begin{tikzpicture}[
    >=Latex,
    x=1cm, y=1cm,
    node_nn/.style={
        circle,
        draw=black!70,
        fill=orange!10,
        minimum size=10mm,
        inner sep=0pt,
        line width=0.55pt,
        font=\small
    },
    node_kan/.style={
        circle,
        draw=black!70,
        fill=blue!8,
        minimum size=10mm,
        inner sep=0pt,
        line width=0.55pt,
        font=\small
    },
    edge_nn/.style={
        draw=black!50,
        line width=0.7pt,
        ->
    },
    edge_kan/.style={
        draw=blue!60!black,
        line width=0.9pt,
        ->
    },
    layer_label/.style={
        font=\small,
        text=black!55
    },
    title_style/.style={
        font=\normalsize\bfseries
    }
]


\node[title_style] at (-3, 7.2) {Feedforward NN};

\node[node_nn] (nn_x1) at (-4.5, 0) {$x_1$};
\node[node_nn] (nn_x2) at (-3.0, 0) {$x_2$};
\node[node_nn] (nn_x3) at (-1.5, 0) {$x_3$};

\node[node_nn] (nn_h1) at (-4.5, 3) {$h_1$};
\node[node_nn] (nn_h2) at (-3.0, 3) {$h_2$};
\node[node_nn] (nn_h3) at (-1.5, 3) {$h_3$};

\node[node_nn] (nn_out) at (-3.0, 6) {$f$};

\node[layer_label, anchor=east] at (-5.6, 0) {Input};
\node[layer_label, anchor=east] at (-5.6, 3) {Hidden};
\node[layer_label, anchor=east] at (-5.6, 6) {Output};

\foreach \inp in {nn_x1, nn_x2, nn_x3} {
    \foreach \hid in {nn_h1, nn_h2, nn_h3} {
        \draw[edge_nn] (\inp) -- (\hid);
    }
}

\foreach \hid in {nn_h1, nn_h2, nn_h3} {
    \draw[edge_nn] (\hid) -- (nn_out);
}

\node[font=\scriptsize, text=orange!70!black]
    at (-4.5, 3) {};
\foreach \hid/\lbl in {nn_h1/$\sigma$, nn_h2/$\sigma$, nn_h3/$\sigma$}{
    \node[font=\scriptsize, text=orange!80!black,
          yshift=-7pt] at (\hid) {\lbl};
}
\node[font=\scriptsize, text=orange!80!black,
      yshift=-7pt] at (nn_out) {$\sigma$};

\node[font=\scriptsize, text=black!50,
      rotate=55] at (-3.95, 1.45) {$w_{ij}$};

\node[font=\footnotesize, align=center, text=black!60]
    at (-3.0, -1.0)
    {Fixed $\sigma$ on nodes\\learned $w_{ij}$ on edges};


\node[title_style] at (3, 7.2) {KAN};

\node[node_kan] (kan_x1) at (1.5, 0) {$x_1$};
\node[node_kan] (kan_x2) at (3.0, 0) {$x_2$};
\node[node_kan] (kan_x3) at (4.5, 0) {$x_3$};

\node[node_kan] (kan_h1) at (1.5, 3) {$\Sigma$};
\node[node_kan] (kan_h2) at (3.0, 3) {$\Sigma$};
\node[node_kan] (kan_h3) at (4.5, 3) {$\Sigma$};

\node[node_kan] (kan_out) at (3.0, 6) {$\Sigma$};

\node[layer_label, anchor=west] at (5.6, 0) {Input};
\node[layer_label, anchor=west] at (5.6, 3) {Hidden};
\node[layer_label, anchor=west] at (5.6, 6) {Output};

\foreach \inp in {kan_x1, kan_x2, kan_x3} {
    \foreach \hid in {kan_h1, kan_h2, kan_h3} {
        \draw[edge_kan] (\inp) -- (\hid);
    }
}

\foreach \hid in {kan_h1, kan_h2, kan_h3} {
    \draw[edge_kan] (\hid) -- (kan_out);
}

\node[font=\scriptsize, text=blue!70!black,
      rotate=55] at (2.5, 1.45) {$\phi_e(\cdot)$};
\node[font=\scriptsize, text=blue!70!black,
      rotate=0] at (3.0, 4.45) {$\phi_e(\cdot)$};

\node[font=\footnotesize, align=center, text=black!60]
    at (3.0, -1.0)
    {Fixed $\Sigma$ on nodes\\learned $\phi_e(\cdot)$ on edges};

\draw[dashed, black!25, line width=0.8pt] (0, -1.6) -- (0, 7.5);

\end{tikzpicture}
\caption{Comparison of a feedforward neural network (left) and a KAN
(right), both with architecture $[3,3,1]$. In the feedforward NN,
nodes apply a fixed nonlinear activation $\sigma$ and edges carry
scalar weights $w_{ij}$. In the KAN, nodes perform simple summation
$\Sigma$ and edges carry learnable univariate functions $\phi_e(\cdot)$.
The key difference is \emph{where} the nonlinearity lives: on the
nodes in a NN, and on the edges in a KAN.}
\label{fig:nn_vs_kan}
\end{figure}

\subsubsection{Theoretical motivation}

The architecture is motivated by the classical Kolmogorov--Arnold
representation theorem, which states that any continuous function
$f:\mathbb{R}^d\to\mathbb{R}$ can be written as a finite composition
of continuous univariate functions and additions:
\[
    f(x_1,\dots,x_d)
    =
    \sum_{q=1}^{2d+1}
    \Phi_q\!\left(
        \sum_{p=1}^{d} \phi_{q,p}(x_p)
    \right).
\]
The theorem does not provide a training algorithm, but it supplies the
conceptual foundation for the KAN architecture: high-dimensional
functions can be captured through compositions of one-dimensional
functions. In this paper we invoke this idea carefully: rather than
relying on the full representation theorem, we work under the
compact-set density assumption stated in
Assumption~\ref{ass:kan_density}, which keeps the theoretical claims
precise without overstating what current theory supports.

\subsubsection{Symbolic extraction}

The key advantage of KANs for our application is \emph{symbolic
extraction}. After training, each learned edge function can be
inspected and, when appropriate, replaced by a known elementary
function. If an edge function closely resembles $\exp(x)$, for
instance, that edge can be locked to the exponential exactly. This
\emph{symbolic locking} step converts the trained KAN into an explicit
algebraic expression built from a chosen library of functions
$\mathcal{L}_{\mathrm{sym}}$.

In our pipeline, symbolic extraction proceeds in stages: the trained
KAN is first pruned to remove weak or redundant edges, the spline grid
is then refined to sharpen the fit, and symbolic locking is applied
only after these preparatory steps. The full procedure is described in
Section~\ref{sec:methodology}. The end result is a closed-form
symbolic surrogate that can be evaluated immediately, inspected for
financial consistency, and tested against known structural properties
of the CAT bond pricing function.

\subsection{Main Contributions}
\label{subsec:contributions}

This paper makes the following contributions.

\begin{enumerate}[label=(\roman*)]
    \item We develop a baseline-plus-residual KAN model for CAT
          bond pricing under lognormal losses. The network learns only
          the deviation from a closed-form lognormal baseline price,
          rather than reconstructing the full price from scratch. This
          decomposition reduces the learning problem to a small, smooth
          residual, making symbolic extraction substantially easier.

    \item We extract a compact symbolic pricing formula from the
          trained KAN and assess its out-of-sample accuracy on a
          held-out test set and on a fully disjoint holdout sample of
          $90{,}000$ observations, obtaining an average relative
          pricing error below $0.5\%$.

    \item We prove monotonicity properties of the true CAT bond price
          with respect to the catastrophe arrival intensity $\lambda$,
          the initial interest rate $r_0$, and the trigger threshold $D$,
          extending the continuity analysis of \cite{CATbondPaper} to
          directional comparative statics.

    \item We derive sufficient conditions on KAN edge functions under
          which the learned model inherits these monotonicity
          properties, providing a theoretical basis for
          structure-aware KAN design in this financial setting.

    \item We formulate a monotonicity-constrained training objective
          via an exterior penalty method and show that the penalised
          problem converges to the constrained problem as the penalty
          weights tend to infinity.

    \item We verify empirically, through sensitivity analysis, that
          the extracted symbolic formula exhibits the expected
          comparative statics with respect to $\lambda$, $r_0$, and
          $D$ across the parameter ranges studied in this paper.
\end{enumerate}

\subsection{Organisation of the Paper}
\label{subsec:organisation}

The remainder of this paper is organised as follows.
Section~\ref{sec:setup} introduces the CAT bond pricing framework:
the compound Poisson loss model, the Vasicek interest rate model,
the resulting pricing formulas, and the Monte Carlo data-generation
scheme with importance sampling.
Section~\ref{sec:theory} presents the theoretical analysis, covering
KAN approximation of the pricing map, monotonicity results, sufficient
conditions for monotonicity-preserving KANs, and the penalised
training formulation.
Section~\ref{sec:methodology} describes the modelling pipeline,
including the baseline-plus-residual setup, hyperparameter search, and
symbolic extraction procedure.
Section~\ref{sec:numerics} reports the numerical results, including
the extracted symbolic formula, out-of-sample accuracy and sensitivity
analysis.


\section{CAT Bond Pricing Model}
\label{sec:setup}

This section summarises the CAT bond pricing framework used throughout
the paper, following \cite{CATbondPaper}. We introduce the stochastic
models for catastrophe losses and interest rates, state the
risk-neutral pricing formulas, and explain how the training data are
generated by Monte Carlo simulation with importance sampling.

\subsection{Loss Process}
\label{subsec:loss}

Let $(\Omega, \mathcal{G}, \mathbb{Q})$ be a probability space, where
$\mathbb{Q}$ is the risk-neutral pricing measure, and let
$\mathbb{F} = \{\mathcal{F}_t\}_{t\ge0}$ be the filtration describing
the flow of market information. Catastrophe arrivals are modeled by a
counting process $\{M(t)\}_{t\ge0}$, where $M(t)$ records the number
of events that have occurred by time $t$.
The quantity $\lambda(t)$ represents the
instantaneous rate of catastrophe arrivals at time $t$. In this paper,
and in line with \cite{CATbondPaper}, we restrict attention to the
constant-intensity case $\lambda(t)\equiv\lambda>0$.

Each catastrophe $i$ generates a loss $X_i>0$, where the severities
$\{X_i\}_{i\ge1}$ are i.i.d.\ with common distribution function $F_X$
and are independent of the counting process $M$. The cumulative
indemnity loss up to time $t$ is therefore the compound Poisson sum
\begin{equation}
    L(t) = \sum_{i=1}^{M(t)} X_i.
    \label{eq:aggregate_loss}
\end{equation}
Throughout the paper we focus on lognormal loss severities,
\[
    X_i \sim \mathrm{Lognormal}(\mu_{\mathrm{sev}},
    \sigma_{\mathrm{sev}}^2),
\]
with parameters $\mu_{\mathrm{sev}} = 18.4$ and
$\sigma_{\mathrm{sev}} = 1.0$, following the empirical calibration in
\cite{BurneckiKukla}. These values imply a mean loss of
approximately $\$1.6\times 10^8$ per event, which is consistent with
large-scale catastrophe losses.

\subsection{Trigger Mechanism}
\label{subsec:trigger}

A trigger event occurs when cumulative losses first reach or exceed a
fixed threshold $D>0$. The trigger time is the stopping time
\[
    \tau = \inf\{t \ge 0 : L(t) \ge D\}.
\]
Equivalently, one may work with the trigger indicator
\[
    N(t) = \mathbf{1}_{\{L(t) \ge D\}}
         = \mathbf{1}_{\{\tau \le t\}},
\]
so that $N(t)=1$ if a trigger has occurred by time $t$, and $N(t)=0$
otherwise.

A key structural result from \cite{CATbondPaper} is that
$\{N(t)\}_{t\ge0}$ is itself a doubly stochastic Poisson process with
$\mathbb{F}$-intensity
\[
    \lambda(t)\bigl(1-F_X(D-L(t))\bigr)\mathbf{1}_{\{L(t)<D\}}.
\]
The intensity is nonzero only while cumulative losses remain below the
trigger threshold. This representation is the main ingredient behind
the pricing formulas stated below.

\subsection{Interest Rate Model}
\label{subsec:rates}

The risk-free interest rate $\{r_t\}_{t\ge0}$ is modeled by the Vasicek
dynamics \cite{Vasicek1977}
\[
    dr_t = \kappa(\theta - r_t)\,dt + \sigma_r\,dW_t,
\]
where $W_t$ is a standard Brownian motion under $\mathbb{Q}$,
$\kappa>0$ is the mean-reversion speed, $\theta>0$ is the long-run
mean level, and $\sigma_r>0$ is the short-rate volatility. The
interest-rate process is assumed to be independent of the catastrophe
loss process $L$.

Under the Vasicek model, the time-0 price of a default-free
zero-coupon bond paying 1 at maturity $t$ is available in closed form:
\begin{equation}
    P_Z(1,0,t) = \exp\!\bigl(A(0,t) - B(0,t)\,r_0\bigr),
    \label{eq:vasicek_bond}
\end{equation}
where $r_0$ is the initial short rate and
\[
    B(0,t) = \frac{1-e^{-\kappa t}}{\kappa},
    \qquad
    A(0,t) = \left(\theta - \frac{\sigma_r^2}{2\kappa^2}\right)
    \bigl(B(0,t) - t\bigr)
    - \frac{\sigma_r^2}{4\kappa}\,B(0,t)^2.
\]
We use the Vasicek model because it gives a tractable discount factor,
which enters both the exact pricing formulas and the closed-form
baseline introduced later in Section~\ref{sec:methodology}.

\subsection{CAT Bond Pricing Formulas}
\label{subsec:pricing}

A \emph{zero-coupon CAT bond} with face value $F$ and maturity $T$
pays $F$ at time $T$ if no trigger event has occurred by that date,
and pays nothing otherwise. A \emph{coupon CAT bond} additionally pays
coupons $C_1,\dots,C_n$ at dates $t_1 < \cdots < t_n \le T$,
conditional on survival up to the corresponding payment dates.

The following pricing relations are stated as Corollary~2.4 in
\cite{CATbondPaper}.

\begin{corollary}[CAT bond pricing formulas, {\cite{CATbondPaper}}]
\label{cor:pricing}
Let $t \le T$.
\begin{enumerate}[label=(\roman*)]
    \item The risk-neutral price of the zero-coupon CAT bond is
    \begin{equation}
        C_Z(F,t,T)
        =
        \Bigl(1 - \mathbb{E}^{\mathbb{Q}}[N(T)\mid\mathcal{F}_t]\Bigr)
        P_Z(F,t,T).
        \label{eq:zero_coupon_price}
    \end{equation}

    \item The risk-neutral price of the coupon CAT bond is
    \begin{equation}
        C_B(F,\{C_i\}_{i=1}^n,t,T)
        =
        C_Z(F,t,T)
        +
        \sum_{i=1}^n \mathbf{1}_{\{t \le t_i\}}\,
        C_Z(C_i,t,t_i).
        \label{eq:coupon_price}
    \end{equation}
\end{enumerate}
\end{corollary}

Formula \eqref{eq:zero_coupon_price} has a simple interpretation: the
CAT bond price equals the default-free discount factor $P_Z(F,t,T)$
multiplied by the survival probability
\[
    1-\mathbb{E}^{\mathbb{Q}}[N(T)\mid\mathcal{F}_t]
    =
    \mathbb{Q}(\tau > T \mid \mathcal{F}_t).
\]
Formula \eqref{eq:coupon_price} then writes the coupon bond price as a
sum of zero-coupon CAT bond prices, one for each contractual
cashflow.

Because the compound Poisson model does not admit a closed-form
expression for $\mathbb{E}^{\mathbb{Q}}[N(T)]$, these formulas must be
evaluated numerically. We therefore describe the Monte Carlo
procedure used to generate the training data.

\subsection{Monte Carlo Data Generation with Importance Sampling}
\label{subsec:mc}

The dataset is generated by evaluating the pricing formulas in
Corollary~\ref{cor:pricing} with Monte Carlo simulation. For a fixed
parameter vector $(r_0,\lambda,D,N,T)$, the zero-coupon price is
estimated by
\[
    \hat{C}_Z = \bigl(1 - \hat{\theta}\bigr) P_Z(F,0,T),
\]
where $\hat{\theta}$ is an estimator of the trigger probability
$\mathbb{Q}(L(T)\ge D)$.

\subsubsection{Naive Monte Carlo}

The most direct estimator simulates $n_s$ independent realisations of
the aggregate loss $L(T)$ and sets
\[
    \hat{\theta}^{\mathrm{MC}}
    = \frac{1}{n_s}\sum_{i=1}^{n_s} \mathbf{1}_{\{L_i(T) \ge D\}}.
\]
This estimator is simple but inefficient when $D$ is large, because
the trigger event becomes rare and only a small fraction of simulated
paths contribute to the estimate. The resulting variance can be large.

\subsubsection{Importance sampling for lognormal losses}

To reduce variance, \cite{CATbondPaper} uses importance sampling
(IS). The main idea is to simulate under a modified measure
$\mathbb{Q}^{a,b}$ that makes threshold exceedance more likely, and
then correct the estimator with an appropriate likelihood ratio.

Under $\mathbb{Q}^{a,b}$, the event count $M(T)$ is Poisson with
intensity $\lambda T e^a$, and the severities $X_i$ are lognormal with
shifted mean parameter $\mu_{\mathrm{sev}}+b$. The resulting IS
estimator is
\begin{equation}
    \hat{\theta}^{\mathrm{IS}}
    = \frac{1}{n_s}\sum_{i=1}^{n_s}
    \mathbf{1}_{\{L_i(T) \ge D\}}\,R(L_i(T)),
    \label{eq:is_estimator}
\end{equation}
where $R(L(T))$ is the likelihood ratio between $\mathbb{Q}$ and
$\mathbb{Q}^{a,b}$:
\begin{equation}
    R(L(T))
    =
    \exp\!\bigl(\lambda T(e^a-1) - aM(T)\bigr)
    \exp\!\left(
        \frac{b^2}{2\sigma_{\mathrm{sev}}^2}M(T)
    \right)
    \exp\!\left(
        -\frac{b}{\sigma_{\mathrm{sev}}^2}
        \sum_{i=1}^{M(T)}(\ln X_i - \mu_{\mathrm{sev}})
    \right).
    \label{eq:likelihood_ratio}
\end{equation}
The parameters $a$ and $b$ are chosen so as to reduce estimator
variance by moving the mean of $L(T)$ toward the threshold $D$. In
practice, they are computed numerically using the characterisation
given in Proposition~3.2 of \cite{CATbondPaper}. Under this scheme,
$\hat{\theta}^{\mathrm{IS}}$ remains unbiased and achieves strictly
smaller variance than naive Monte Carlo for sufficiently large
thresholds \cite{CATbondPaper}.

\subsubsection{Dataset}

Using the procedure above, we generate a dataset of $100{,}000$ CAT
bond prices over the parameter ranges reported in
Table~\ref{tab:data_params_sec2}, following the empirical calibration
in \cite{CATbondPaper}.

\begin{table}[ht]
\centering
\caption{Parameter ranges for training data generation.}
\label{tab:data_params_sec2}
\begin{tabular}{ll}
\toprule
\textbf{Parameter} & \textbf{Value / Range} \\
\midrule
Coupon rate $c$          & $5\%$ \\
Vasicek parameters       & $\kappa=0.2$, $\theta=3\%$, $\sigma_r=2\%$ \\
Initial short rate $r_0$ & $\mathrm{Uniform}(0,\,0.08)$ \\
Catastrophe intensity $\lambda$
                         & $\mathrm{Uniform}(30,\,40)$ \\
Trigger threshold $D$    & $\mathrm{Uniform}(7\times10^9,\,
                           13\times10^9)$ \\
Coupon frequency $N$     & $\{0,2,3,4,6,8,10,12\}$ \\
Maturity $T$             & $\mathrm{Uniform}(90,\,720)$ days \\
\bottomrule
\end{tabular}
\end{table}

For each observation, the features $(r_0,\lambda,D,N,T)$ are sampled
from the distributions in Table~\ref{tab:data_params_sec2}, and the
corresponding CAT bond price is computed by the IS-enhanced Monte
Carlo procedure. The resulting $100{,}000$ observations are used for
all subsequent training and evaluation.

\subsection{Baseline-Plus-Residual Learning Setup}
\label{subsec:setup_residual}

A key design choice in our approach is to train the KAN not on the 
raw price $P$ but on a residual relative to a closed-form baseline 
price $\Pbase$. If the baseline captures most of the pricing 
structure, the residual is small and smooth, which substantially 
reduces the complexity of the learning problem. It also encourages 
the extracted symbolic formula to represent financially interpretable 
corrections to the baseline rather than relearning the full pricing 
structure from scratch. The full mathematical construction is given 
in Section~\ref{sec:methodology}.

\section{Theoretical Analysis}
\label{sec:theory}
 
This section develops the theoretical foundations of the paper in four
parts. We first establish that the CAT bond pricing map can be
approximated arbitrarily well by KANs on compact domains. We then prove that the true CAT bond price
satisfies natural monotonicity properties with respect to the
event intensity $\lambda$, the initial interest rate $r_0$, and the
trigger threshold $D$. Next, we derive sufficient conditions on KAN
edge functions that guarantee these monotonicities are preserved by the
learned model. Finally, we formulate a penalised training objective
that enforces the monotonicity constraints and prove that it
converges to the constrained problem as the penalty weights tend to
infinity.
 
Throughout this section, we work at time $t=0$ with constant event
intensity $\lambda(t)\equiv\lambda$, and write
$C_Z(F,T;r_0,\lambda,D):=C_Z(F,0,T)$ to emphasise dependence on the
key pricing inputs.

\subsection{KAN Approximation of the CAT Bond Pricing Map}
\label{subsec:kan_approx}
 
Sester and Xu \cite{CATbondPaper} prove that the pricing maps $C_Z$ and $C_B$ can be
approximated arbitrarily well by fully connected feedforward neural
networks on compact domains, using two ingredients: continuity of the
pricing map, and the universal approximation property of the network
class. We follow exactly the same strategy here, replacing the
feedforward network class by a suitable class of KANs.
 
To avoid overstating the currently available KAN approximation theory,
we formulate the density step through an assumption. This
keeps the subsequent propositions clean and allows the results to be
updated as KAN theory develops.
 
\begin{assumption}[KAN density on compact sets]
\label{ass:kan_density}
For each input dimension $d \in \mathbb{N}$, let $\mathcal{K}_{d,1}$
denote the chosen class of scalar-valued KANs on $\mathbb{R}^d$ (for
example, spline-based KANs with sufficiently rich width, depth, and
grid resolution). Assume that for every compact set
$K \subset \mathbb{R}^d$, the restricted class
$\mathcal{K}_{d,1}|_{K}$ is dense in $C(K,\mathbb{R})$ with respect
to the uniform norm. Equivalently, for every $f \in C(K,\mathbb{R})$
and every $\varepsilon > 0$, there exists
$f_{\mathrm{KAN}} \in \mathcal{K}_{d,1}$ such that
\[
    \sup_{x \in K} |f_{\mathrm{KAN}}(x) - f(x)| < \varepsilon.
\]
\end{assumption}
 
\begin{remark}
Assumption~\ref{ass:kan_density} is the KAN analogue of the universal
approximation property used in Proposition~4.1 of \cite{CATbondPaper}.
It serves here as an approximation hypothesis. Existing KAN
approximation results, such as Theorem~2.1 in \cite{KANpaper}, instead
establish approximation under additional smooth compositional
representation assumptions on the target map.
\end{remark}
 
\begin{proposition}[KAN approximation of the zero-coupon CAT bond
pricing map]
\label{prop:kan_approx_zero}
Let Assumption~\ref{ass:kan_density} hold for $d=5$, and let
$K \subset \mathbb{R}^{5}$ be compact, with coordinates
$k=(r_0,\lambda,T,F,D)$. Assume that the zero-coupon CAT bond pricing
map $f(k):=C_Z(F,T;r_0,\lambda,D)$ is continuous on $K$. Then for
every $\varepsilon>0$, there exists $f_{\mathrm{KAN}} \in
\mathcal{K}_{5,1}$ such that
\[
    \sup_{k\in K}\bigl|f_{\mathrm{KAN}}(k)
    -C_Z(F,T;r_0,\lambda,D)\bigr|<\varepsilon.
\]
\end{proposition}
 
\begin{proof}
Since $f$ is continuous on the compact set $K$, we have
$f \in C(K,\mathbb{R})$. Applying Assumption~\ref{ass:kan_density}
with $d=5$ yields that for every $\varepsilon>0$ there exists
$f_{\mathrm{KAN}} \in \mathcal{K}_{5,1}$ satisfying
$\sup_{k\in K}|f_{\mathrm{KAN}}(k)-f(k)|<\varepsilon$, which is the
desired conclusion.
\end{proof}
 
\begin{proposition}[KAN approximation of the coupon CAT bond pricing
map]
\label{prop:kan_approx_coupon}
Let $n\in\mathbb{N}$ be fixed. Suppose Assumption~\ref{ass:kan_density}
holds for $d=2n+5$, and let $K_B \subset \mathbb{R}^{2n+5}$ be
compact, with coordinates
$k_B=(r_0,\lambda,\{t_i\}_{i=1}^n,T,\{C_i\}_{i=1}^n,F,D)$. Assume
that the coupon CAT bond pricing map
$f_B(k_B):=C_B(F,\{C_i\}_{i=1}^n,\{t_i\}_{i=1}^n,T;r_0,\lambda,D)$
is continuous on $K_B$. Then for every $\varepsilon>0$, there exists
$f_{B,\mathrm{KAN}} \in \mathcal{K}_{2n+5,1}$ such that
\[
    \sup_{k_B\in K_B}\bigl|f_{B,\mathrm{KAN}}(k_B)
    -C_B(F,\{C_i\}_{i=1}^n,\{t_i\}_{i=1}^n,T;r_0,\lambda,D)
    \bigr|<\varepsilon.
\]
\end{proposition}
 
\begin{proof}
By assumption $f_B \in C(K_B,\mathbb{R})$. Applying
Assumption~\ref{ass:kan_density} with $d=2n+5$ gives the result
directly. Continuity of $f_B$ may alternatively be verified from the
decomposition
\[
    C_B = C_Z(F,T;r_0,\lambda,D)
    +\sum_{i=1}^n C_Z(C_i,t_i;r_0,\lambda,D),
\]
since each zero-coupon term is continuous by Proposition~4.2 of
\cite{CATbondPaper}, and a finite sum of continuous functions is
continuous.
\end{proof}
 
\begin{remark}
Propositions~\ref{prop:kan_approx_zero} and
\ref{prop:kan_approx_coupon} are direct KAN analogues of
Propositions~4.2 and 4.3 in \cite{CATbondPaper}. They confirm that
replacing feedforward networks by KANs does not sacrifice the
theoretical approximation guarantee, while KANs additionally offer a
route toward symbolic extraction and interpretable pricing formulas.
\end{remark}
 
\subsection{Structural Monotonicity of the True CAT Bond Price}
\label{subsec:monotonicity}
 
We now establish that the true CAT bond pricing map satisfies natural
monotonicity properties that are consistent with financial intuition:
higher catastrophe intensity and higher interest rates reduce the bond
price, while a higher trigger threshold increases it.
 
\subsubsection{Zero-coupon CAT bond}
 
\begin{theorem}[Monotonicity of the zero-coupon CAT bond price]
\label{thm:zero_coupon_monotonicity}
Assume:
\begin{enumerate}
    \item[(a)] the severities $X_i$ satisfy $X_i\ge 0$ almost surely;
    \item[(b)] the aggregate loss process $L$ is independent of the
               short-rate process under $\mathbb{Q}$;
    \item[(c)] the default-free bond price $P_Z(F,0,T;r_0)$ is
               nonincreasing in $r_0$;
    \item[(d)] the Vasicek bond price satisfies
               $P_Z(F,0,T)=F\exp\bigl(A(0,T)-B(0,T)r_0\bigr)$
               with $B(0,T)\ge0$.
\end{enumerate}
Then the zero-coupon CAT bond price
$(r_0,\lambda,D)\mapsto C_Z(F,T;r_0,\lambda,D)$ is
\begin{itemize}
    \item nonincreasing in $\lambda$,
    \item nonincreasing in $r_0$,
    \item nondecreasing in $D$.
\end{itemize}
\end{theorem}
 
\begin{proof}
From the pricing formula in Section~2, we have the factorisation
\begin{equation}
    C_Z(F,T;r_0,\lambda,D)
    =
    \mathbb{Q}(L(T)<D)\,P_Z(F,0,T;r_0),
    \label{eq:factorisation}
\end{equation}
since $1-\mathbb{E}^{\mathbb{Q}}[N(T)]=\mathbb{Q}(\tau>T)=\mathbb{Q}(L(T)<D)$
and the loss and short-rate processes are independent by assumption~(b).
 
\medskip
\noindent\textbf{Monotonicity in $r_0$.}
The survival probability $\mathbb{Q}(L(T)<D)$ does not depend on
$r_0$ by assumption~(b). Hence from \eqref{eq:factorisation},
$C_Z = \mathbb{Q}(L(T)<D)\,P_Z(F,0,T;r_0)$ is nonincreasing in
$r_0$ whenever $P_Z$ is nonincreasing in $r_0$, which holds by
assumption~(c). Under the Vasicek specification~(d), this follows
explicitly by differentiation:
\[
    \frac{\partial C_Z}{\partial r_0}
    =
    -B(0,T)\,\mathbb{Q}(L(T)<D)\,
    F\exp\bigl(A(0,T)-B(0,T)r_0\bigr).
\]
Since $B(0,T)\ge0$, $\mathbb{Q}(L(T)<D)\ge0$, and the exponential is
strictly positive, the derivative is nonpositive.
 
\medskip
\noindent\textbf{Monotonicity in $D$.}
Fix $r_0$, $\lambda$, $T$ and let $D_1\le D_2$. Then pathwise
$\{L(T)<D_1\}\subseteq\{L(T)<D_2\}$, so
\[
    \mathbb{Q}(L(T)<D_1)\le\mathbb{Q}(L(T)<D_2).
\]
Multiplying by the strictly positive factor $P_Z(F,0,T;r_0)$ gives
$C_Z(F,T;r_0,\lambda,D_1)\le C_Z(F,T;r_0,\lambda,D_2)$.
 
\medskip
\noindent\textbf{Monotonicity in $\lambda$.}
Fix $r_0$, $D$, $T$ and let $0<\lambda_1<\lambda_2$. We construct a
coupling on a common probability space. Let $M^{(1)}(t)$ be a Poisson
process with intensity $\lambda_1$, and let $\widetilde{M}(t)$ be an
independent Poisson process with intensity $\lambda_2-\lambda_1$.
Define
\[
    M^{(2)}(t) := M^{(1)}(t)+\widetilde{M}(t),
\]
so that $M^{(2)}$ is Poisson with intensity $\lambda_2$ and
$M^{(2)}(t)\ge M^{(1)}(t)$ pathwise. Let $(X_i)_{i\ge1}$ and
$(\widetilde{X}_i)_{i\ge1}$ be i.i.d.\ nonnegative severity sequences
independent of the counting processes, and define
\[
    L^{(k)}(t):=\sum_{i=1}^{M^{(k)}(t)}X_i, \quad k=1,2,
\]
where $L^{(2)}$ uses the additional severities for the extra jumps of
$\widetilde{M}$. Since all severities are nonnegative almost surely,
\[
    L^{(2)}(t)\ge L^{(1)}(t) \quad \text{pathwise},
\]
which gives the set inclusion $\{L^{(2)}(T)<D\}\subseteq\{L^{(1)}(T)<D\}$
and therefore
\[
    \mathbb{Q}_{\lambda_2}(L(T)<D)
    \le
    \mathbb{Q}_{\lambda_1}(L(T)<D).
\]
Multiplying by $P_Z(F,0,T;r_0)$, which does not depend on $\lambda$,
gives $C_Z(F,T;r_0,\lambda_2,D)\le C_Z(F,T;r_0,\lambda_1,D)$.
 
This completes the proof of all three monotonicity statements.
\end{proof}
 
\begin{remark}
\label{rem:decomposition}
The factorisation \eqref{eq:factorisation} reveals a useful structural
decomposition: the CAT bond price separates into a
\emph{trigger-survival component} $\mathbb{Q}(L(T)<D)$, driven by
$\lambda$ and $D$, and a \emph{default-free discounting component}
$P_Z(F,0,T;r_0)$, driven by $r_0$. This separation underlies both the
monotonicity proofs and the baseline-plus-residual learning formulation
in Section~4.
\end{remark}
 
\subsubsection{Coupon CAT bond}
 
\begin{corollary}[Monotonicity of the coupon CAT bond price]
\label{cor:coupon_monotonicity}
Under the assumptions of Theorem~\ref{thm:zero_coupon_monotonicity},
the coupon CAT bond price
\[
    C_B(F,\{C_i\}_{i=1}^n,\{t_i\}_{i=1}^n,T;r_0,\lambda,D)
    =
    C_Z(F,0,T;r_0,\lambda,D)
    +\sum_{i=1}^n C_Z(C_i,0,t_i;r_0,\lambda,D)
\]
is nonincreasing in $\lambda$, nonincreasing in $r_0$, and nondecreasing
in $D$.
\end{corollary}
 
\begin{proof}
The coupon CAT bond price is a finite sum of zero-coupon CAT bond
prices. Each summand satisfies the stated monotonicity properties by
Theorem~\ref{thm:zero_coupon_monotonicity}, and a finite sum of
functions with the same monotonicity direction preserves that
monotonicity.
\end{proof}
 
\begin{remark}
Theorem~\ref{thm:zero_coupon_monotonicity} and
Corollary~\ref{cor:coupon_monotonicity} go beyond the continuity
results in \cite{CATbondPaper} by establishing directional structure
in the pricing map. In symbolic regression, two formulas with similar
out-of-sample error may differ in financial credibility if one violates
these comparative statics while the other respects them. The
monotonicity results therefore provide a principled secondary model
selection criterion, complementing predictive accuracy.
\end{remark}
 
\subsection{Sufficient Conditions for Monotonicity-Preserving KANs}
\label{subsec:mono_kan}
 
We now identify simple structural conditions under which a KAN
inherits the monotonicity of the true pricing map.

\begin{theorem}[Sufficient conditions for coordinatewise monotonicity]
\label{thm:kan_monotone}
Let $f:\mathbb{R}^d\to\mathbb{R}$ be a scalar-output KAN in which each
hidden node computes a sum of its incoming edge outputs. Fix a
coordinate $x_j$.

Assume that every edge function $\phi_e:\mathbb{R}\to\mathbb{R}$ lying
on a directed path from input $x_j$ to the output is nondecreasing on
the domain of interest. Then $f$ is nondecreasing in $x_j$.

Similarly, if every such edge function is nonincreasing, then $f$ is
nonincreasing in $x_j$.
\end{theorem}

\begin{proof}
We prove the nondecreasing case.
For each node $v$, let $z_v$ denote its output. We show by induction
on the layers that every node reachable from $x_j$ has output
nondecreasing in $x_j$.

At the input layer, this is immediate since $x_j$ is nondecreasing
in itself.

Now let $v$ be a node in some later layer and assume that every
predecessor $u$ of $v$ lying on a directed path from $x_j$ has
output $z_u$ nondecreasing in $x_j$. Since $v$ computes the sum of
its incoming edge outputs,
\[
    z_v = \sum_{u \to v} \phi_{uv}(z_u).
\]
For edges $u \to v$ lying on a directed path from $x_j$, the function
$\phi_{uv}$ is nondecreasing by assumption and $z_u$ is nondecreasing
in $x_j$ by the induction hypothesis, so $\phi_{uv}(z_u)$ is
nondecreasing in $x_j$. For edges $u \to v$ not lying on any directed
path from $x_j$, the value $z_u$ does not depend on $x_j$, so
$\phi_{uv}(z_u)$ is constant in $x_j$. A sum of nondecreasing and
constant functions is nondecreasing. Hence $z_v$ is nondecreasing in
$x_j$.

By induction, the output node is nondecreasing in $x_j$, so $f$ is
nondecreasing in $x_j$.
\end{proof}

\begin{remark}
Theorem~\ref{thm:kan_monotone} gives a simple sufficient condition for
coordinatewise monotonicity in a KAN. It does not characterise all
monotone KANs. Instead, it shows that monotonicity can be guaranteed by
construction by restricting the edge functions along paths from a
chosen input coordinate to the output.
\end{remark}

\begin{figure}[H]
\centering
\begin{tikzpicture}[
    >=Latex,
    x=1cm, y=1cm,
    node_style/.style={
        circle,
        draw=black!70,
        fill=blue!5,
        minimum size=8.8mm,
        inner sep=0pt,
        line width=0.55pt,
        font=\small
    },
    edge_grey/.style={
        draw=gray!28,
        line width=0.45pt
    },
    edge_red/.style={
        draw=red!75!black,
        line width=0.9pt
    },
    layer_label/.style={
        font=\small
    }
]

\node[node_style] (r)      at (-2.4,0) {$r$};
\node[node_style] (lam)    at (-1.2,0) {$\lambda$};
\node[node_style] (logD)   at (0,0) {$\log D$};
\node[node_style] (N)      at (1.2,0) {$N$};
\node[node_style] (T)      at (2.4,0) {$T$};

\node[node_style] (h1) at (-3.0,2.8) {$h_1$};
\node[node_style] (h2) at (-1.8,2.8) {$h_2$};
\node[node_style] (h3) at (-0.6,2.8) {$h_3$};
\node[node_style] (h4) at (0.6,2.8) {$h_4$};
\node[node_style] (h5) at (1.8,2.8) {$h_5$};
\node[node_style] (h6) at (3.0,2.8) {$h_6$};

\node[node_style] (out) at (0,5.8) {$\hat C$};

\node[layer_label] at (0,-0.82) {Input layer};
\node[layer_label] at (0,3.55) {Hidden layer};
\node[layer_label] at (0,6.45) {Output layer};

\foreach \inp in {r,lam,logD,N,T} {
    \foreach \hid in {h1,h2,h3,h4,h5,h6} {
        \draw[edge_grey,->] (\inp) -- (\hid);
    }
}

\foreach \hid in {h1,h2,h3,h4,h5,h6} {
    \draw[edge_grey,->] (\hid) -- (out);
}

\foreach \hid in {h1,h2,h3,h4,h5,h6} {
    \draw[edge_red,->] (lam) -- (\hid);
    \draw[edge_red,->] (\hid) -- (out);
}

\node[font=\scriptsize,text=red!70!black] at (-1.95,1.22) {$\phi_e$};
\node[font=\scriptsize,text=red!70!black] at (-0.95,1.52) {$\phi_e$};
\node[font=\scriptsize,text=red!70!black] at (0.92,4.42) {$\phi_e$};

\end{tikzpicture}
\caption{Illustration of Theorem~\ref{thm:kan_monotone} for a KAN 
with architecture $[5,6,1]$. The highlighted edges are precisely 
those lying on directed paths from the input $\lambda$ to the 
output $\hat{C}$. If every highlighted edge function $\phi_e$ is 
nondecreasing (or\ nonincreasing) on the domain of interest, 
then the output inherits the same monotonicity in $\lambda$.}
\label{fig:kan_monotone_lambda}
\end{figure}
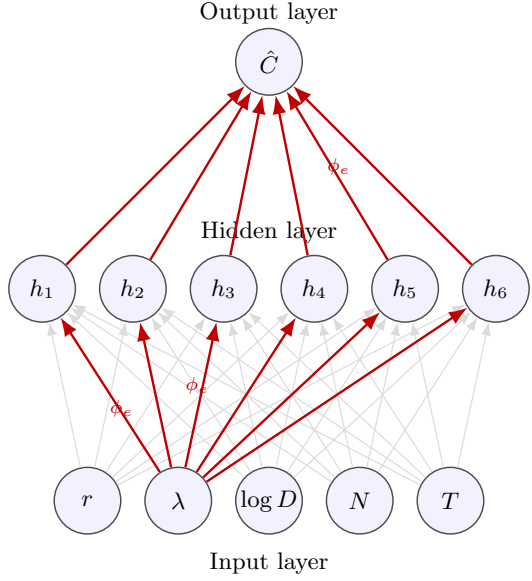

\subsubsection{Application to CAT bond pricing}

The desired monotonicity directions are mixed:
$\partial C/\partial\lambda\le0$,
$\partial C/\partial r_0\le0$,
$\partial C/\partial D\ge0$.
Define
\begin{equation}
    u_1 := -\lambda, \qquad u_2 := -r_0, \qquad u_3 := D,
    \label{eq:reparam}
\end{equation}
so that the reparametrised pricing map
$\widetilde{C}(u_1,u_2,u_3,\ldots)$ is nondecreasing in each of
$u_1$, $u_2$, $u_3$.

\begin{corollary}[Monotonicity-preserving KAN for CAT bond pricing]
\label{cor:mono_kan_catbond}
Let a KAN take inputs
$(u_1,u_2,u_3,\ldots)=(-\lambda,-r_0,D,\ldots)$.
If all edge functions on every directed path from $u_1$, $u_2$, $u_3$
to the output are nondecreasing, then the KAN output is nonincreasing
in $\lambda$, nonincreasing in $r_0$, and nondecreasing in $D$.
\end{corollary}

The result follows immediately from Theorem~\ref{thm:kan_monotone}
by substituting $u_1=-\lambda$, $u_2=-r_0$, $u_3=D$.
 
\subsection{Monotonicity-Constrained KAN Training}
\label{subsec:constrained_training}
 
Corollary~\ref{cor:mono_kan_catbond} provides a hard architectural
guarantee of monotonicity. In practice, however, one may wish to train
an unconstrained KAN and then encourage monotonicity through the
objective function. We formulate a penalised training problem and show
it converges to the constrained problem as the penalty weights grow.
 
\subsubsection{Unconstrained and constrained problems}
 
The unconstrained KAN training problem minimises the empirical mean
squared error:
\begin{equation}
    \inf_{\theta}\,\mathcal{L}_\theta(\mathbf{x}, y)
    :=
    \inf_{\theta}\,\frac{1}{n}\sum_{i=1}^{n}
    \Bigl(\hat{C}_\theta(\mathbf{x}_i) - y_i\Bigr)^2,
    \label{eq:uncon_loss}
\end{equation}
where $\hat{C}_\theta$ is the KAN predictor with parameters $\theta$,
$\mathbf{x}_i$ are the input features, and $y_i$ are the target prices.
The structurally constrained problem is
\begin{equation}
    \inf_{\theta \in \mathcal{F}}\,\mathcal{L}_\theta(\mathbf{x}, y),
    \label{eq:con_prob}
\end{equation}
where the feasible set includes the three monotonicity constraints:
\begin{equation}
    \mathcal{F}
    :=
    \left\{
        \theta
        \;:\;
        \frac{\partial \hat{C}_\theta}{\partial \lambda} \leq 0,
        \quad
        \frac{\partial \hat{C}_\theta}{\partial r_0} \leq 0,
        \quad
        \frac{\partial \hat{C}_\theta}{\partial D} \geq 0
    \right\}.
    \label{eq:feasible_set}
\end{equation}
Note that $\mathcal{F}$ is non-empty: by
Corollary~\ref{cor:mono_kan_catbond}, any KAN satisfying the
sufficient edge-function conditions belong to
$\mathcal{F}$.
 
\subsubsection{Penalised objective}

We add penalty terms to the training objective:
\begin{equation}
    \tilde{\mathcal{L}}_{\theta,\boldsymbol{\mu}}(\mathbf{x}, y)
    :=
    \mathcal{L}_\theta(\mathbf{x}, y)
    +\mu_1\,\mathbb{E}_{\nu}\!\left[
        \Bigl(\tfrac{\partial \hat{C}_\theta}{\partial \lambda}
        \Bigr)^{\!+}
    \right]
    +\mu_2\,\mathbb{E}_{\nu}\!\left[
        \Bigl(\tfrac{\partial \hat{C}_\theta}{\partial r_0}
        \Bigr)^{\!+}
    \right]
    +\mu_3\,\mathbb{E}_{\nu}\!\left[
        \Bigl(-\tfrac{\partial \hat{C}_\theta}{\partial D}
        \Bigr)^{\!+}
    \right],
    \label{eq:penalised_loss}
\end{equation}
where $(z)^+:=\max(z,0)$ and
$\boldsymbol{\mu}=(\mu_1,\mu_2,\mu_3)\in\mathbb{R}_+^3$ are
non-negative penalty weights.

Here, $\nu$ is a probability measure on the input domain. It can conceptualised as a distribution that assigns positive probability to every region
of the input space. The three penalty terms work as follows: the first
penalises any region where the predicted price increases with $\lambda$
(which it should not), the second penalises any region where it
increases with $r_0$ (which it should not), and the third penalises
any region where it decreases with $D$ (which it should not).

The key reason for using an expectation over $\nu$ rather than checking
the constraints at a finite number of points is that monotonicity must
hold everywhere on a continuum of inputs, not just at the sampled
points. If the derivative $\frac{\partial \hat{C}_\theta}{\partial
\lambda}$ is positive at some point, it must also be positive in a
small neighbourhood around that point by continuity. Since $\nu$
assigns positive probability to every such neighbourhood, the
corresponding penalty term will be strictly positive whenever any
violation exists. Conversely, a zero penalty guarantees that no
violation occurs anywhere on the domain.
\subsubsection{Convergence of the penalised problem}

The penalised objective in \eqref{eq:penalised_loss} may be viewed as
an exterior penalty relaxation of the constrained optimisation
problem. Recall that the feasible set is
\[
\mathcal{F}
:=
\left\{
\theta:
\frac{\partial \hat{C}_\theta}{\partial \lambda}\le 0,\quad
\frac{\partial \hat{C}_\theta}{\partial r_0}\le 0,\quad
\frac{\partial \hat{C}_\theta}{\partial D}\ge 0
\ \text{on the domain of interest}
\right\}.
\]

For convenience, define
\[
P_1(\theta)
=
\mathbb{E}_{\nu}\!\left[
\Bigl(\tfrac{\partial \hat{C}_\theta}{\partial \lambda}\Bigr)^+
\right],
\qquad
P_2(\theta)
=
\mathbb{E}_{\nu}\!\left[
\Bigl(\tfrac{\partial \hat{C}_\theta}{\partial r_0}\Bigr)^+
\right],
\qquad
P_3(\theta)
=
\mathbb{E}_{\nu}\!\left[
\Bigl(-\tfrac{\partial \hat{C}_\theta}{\partial D}\Bigr)^+
\right].
\]
Each \(P_j(\theta)\) is non-negative. Under the standing assumption that
the relevant partial derivatives are continuous and that \(\nu\)
assigns positive mass to every nonempty open subset of the domain, we
have
\[
P_1(\theta)=P_2(\theta)=P_3(\theta)=0
\quad\Longleftrightarrow\quad
\theta\in\mathcal{F}.
\]

We also recall two standard definitions. A sequence
\((\theta_n)\) is called a \emph{minimising sequence} for a function
\(J\) if
\[
J(\theta_n)\to \inf_{\theta} J(\theta)
\qquad\text{as }n\to\infty.
\]
A function \(f\) is called \emph{lower semicontinuous} if, whenever
\(\theta_n\to\bar{\theta}\), one has
\[
f(\bar{\theta})
\le
\liminf_{n\to\infty} f(\theta_n).
\]

The following result gives the standard value-level justification for
the penalty approach.

\begin{proposition}[Exterior penalty convergence]
\label{prop:penalty_convergence}
Assume that the feasible set \(\mathcal{F}\) is non-empty. For
\(\boldsymbol{\mu}=(\mu_1,\mu_2,\mu_3)\in\mathbb{R}_+^3\), define
\[
J_{\boldsymbol{\mu}}(\theta)
:=
\mathcal{L}_\theta(\mathbf{x},y)
+\mu_1 P_1(\theta)+\mu_2 P_2(\theta)+\mu_3 P_3(\theta).
\]
Assume further that \(\mathcal{L}_\theta(\mathbf{x},y)\) and each
\(P_j(\theta)\) are lower semicontinuous in \(\theta\), and that every
minimising sequence of the penalised problems admits a convergent
subsequence. Then
\begin{equation}
\lim_{\mu_1,\mu_2,\mu_3\to\infty}
\inf_{\theta} J_{\boldsymbol{\mu}}(\theta)
=
\inf_{\theta\in\mathcal{F}}
\mathcal{L}_\theta(\mathbf{x},y).
\label{eq:penalty_convergence}
\end{equation}
\end{proposition}

\begin{proof}
Set
\[
L^*:=\inf_{\theta\in\mathcal{F}}\mathcal{L}_\theta(\mathbf{x},y),
\]
which is well-defined because \(\mathcal{F}\neq\emptyset\). For each
\(\boldsymbol{\mu}\), define the penalised optimal value
\[
v(\boldsymbol{\mu})
:=
\inf_{\theta} J_{\boldsymbol{\mu}}(\theta).
\]

We show that \(v(\boldsymbol{\mu})\to L^*\) as
\(\mu_1,\mu_2,\mu_3\to\infty\).

\medskip
\noindent
\textbf{Step 1: Upper bound.}
Let \(\theta\in\mathcal{F}\). Then
\(P_1(\theta)=P_2(\theta)=P_3(\theta)=0\), so
\[
J_{\boldsymbol{\mu}}(\theta)=\mathcal{L}_\theta(\mathbf{x},y).
\]
Since this holds for every \(\theta\in\mathcal{F}\), taking the
infimum over \(\mathcal{F}\) yields
\[
v(\boldsymbol{\mu})
=
\inf_{\theta} J_{\boldsymbol{\mu}}(\theta)
\le
\inf_{\theta\in\mathcal{F}}\mathcal{L}_\theta(\mathbf{x},y)
=
L^*.
\]
Hence
\[
v(\boldsymbol{\mu})\le L^*
\qquad\text{for every }\boldsymbol{\mu}.
\]

\medskip
\noindent
\textbf{Step 2: Lower bound.}
We now show that
\[
\liminf_{\mu_1,\mu_2,\mu_3\to\infty} v(\boldsymbol{\mu})\ge L^*.
\]

Suppose, for contradiction, that this is false. Then there exist
\(\varepsilon>0\) and a sequence
\(\boldsymbol{\mu}^{(n)}=(\mu_1^{(n)},\mu_2^{(n)},\mu_3^{(n)})\) such
that
\[
\mu_1^{(n)},\mu_2^{(n)},\mu_3^{(n)}\to\infty
\qquad\text{as }n\to\infty,
\]
and
\[
v(\boldsymbol{\mu}^{(n)})\le L^*-\varepsilon
\qquad\text{for all }n.
\]

By the definition of infimum, for each \(n\) there exists
\(\theta_n\) such that
\[
J_{\boldsymbol{\mu}^{(n)}}(\theta_n)
\le
v(\boldsymbol{\mu}^{(n)})+\frac{1}{n}
\le
L^*-\varepsilon+\frac{1}{n}.
\]
Thus \((\theta_n)\) is a minimising sequence for the penalised
problems.

Since every penalty term is non-negative, we have
\[
\mathcal{L}_{\theta_n}(\mathbf{x},y)
\le
J_{\boldsymbol{\mu}^{(n)}}(\theta_n)
\le
L^*-\varepsilon+\frac{1}{n}.
\]
Also, for each \(j=1,2,3\),
\[
\mu_j^{(n)} P_j(\theta_n)
\le
J_{\boldsymbol{\mu}^{(n)}}(\theta_n)
\le
L^*-\varepsilon+\frac{1}{n}.
\]
Because \(\mu_j^{(n)}\to\infty\), it follows that
\[
P_j(\theta_n)\to 0
\qquad\text{for }j=1,2,3.
\]

By assumption, the minimising sequence \((\theta_n)\) admits a
convergent subsequence. Passing to this subsequence if necessary, we
may assume that
\[
\theta_n\to \bar{\theta}.
\]

Since each \(P_j\) is lower semicontinuous,
\[
P_j(\bar{\theta})
\le
\liminf_{n\to\infty} P_j(\theta_n)
=
0.
\]
As each \(P_j(\bar{\theta})\ge 0\), we obtain
\[
P_j(\bar{\theta})=0
\qquad\text{for }j=1,2,3.
\]
Hence \(\bar{\theta}\in\mathcal{F}\).

Next, by lower semicontinuity of
\(\mathcal{L}_\theta(\mathbf{x},y)\),
\[
\mathcal{L}_{\bar{\theta}}(\mathbf{x},y)
\le
\liminf_{n\to\infty}\mathcal{L}_{\theta_n}(\mathbf{x},y)
\le
L^*-\varepsilon.
\]
But since \(\bar{\theta}\in\mathcal{F}\), the definition of \(L^*\)
gives
\[
L^*
\le
\mathcal{L}_{\bar{\theta}}(\mathbf{x},y).
\]
Combining the last two inequalities yields
\[
L^*
\le
\mathcal{L}_{\bar{\theta}}(\mathbf{x},y)
\le
L^*-\varepsilon,
\]
which is impossible. This contradiction proves that
\[
\liminf_{\mu_1,\mu_2,\mu_3\to\infty} v(\boldsymbol{\mu})\ge L^*.
\]

\medskip
\noindent
\textbf{Step 3: Conclusion.}
From Step 1 we have \(v(\boldsymbol{\mu})\le L^*\) for every
\(\boldsymbol{\mu}\), so
\[
\limsup_{\mu_1,\mu_2,\mu_3\to\infty} v(\boldsymbol{\mu})\le L^*.
\]
From Step 2 we have
\[
\liminf_{\mu_1,\mu_2,\mu_3\to\infty} v(\boldsymbol{\mu})\ge L^*.
\]
Therefore,
\[
\lim_{\mu_1,\mu_2,\mu_3\to\infty} v(\boldsymbol{\mu})=L^*,
\]
which is exactly \eqref{eq:penalty_convergence}.
\end{proof}
\section{KAN Model Design and Symbolic Extraction}
\label{sec:methodology}
 
This section describes the full modeling pipeline used to construct a
symbolic KAN surrogate for CAT bond pricing. The pipeline consists of
five stages: (i) computing a closed-form baseline price for each
observation; (ii) defining a log-ratio residual target and
standardising inputs and outputs; (iii) selecting KAN hyperparameters
via Tree-structured Parzen Estimation (TPE); (iv) extracting a
symbolic formula through pruning, grid refinement, and symbolic
locking; and (v) evaluating the final model on a strictly held-out
test set and a disjoint holdout sample. Each stage is described in
detail below.
 
\subsection{Closed-Form Baseline Price}
\label{subsec:baseline}
 
We work with the dataset of $100{,}000$ simulated CAT bond prices
described in Section~\ref{sec:setup}. From this dataset we draw a
working sample of $10{,}000$ observations, partitioned into training
($70\%$), validation ($15\%$), and test ($15\%$) splits using a fixed
random seed, yielding $7{,}000$ training, $1{,}500$ validation, and
$1{,}500$ test observations. The validation set is used exclusively
during hyperparameter selection; the test set is held out entirely
until final model evaluation.
 
Rather than fitting the raw CAT bond price directly, we exploit the
known financial structure of the pricing problem by first constructing
a closed-form baseline price $\Pbase$, and then training the KAN on
the residual deviation from this baseline. This decomposition
substantially reduces the magnitude and complexity of the learning
target, and encourages the symbolic formula to represent only the
correction to the financial prior rather than relearning discounting
and survival effects already available in closed form.
 
\subsubsection{Vasicek discount factor}
 
The default-free discount factor $P(0,t)$ is computed using the
closed-form Vasicek bond price given in
equation~\eqref{eq:vasicek_bond}.
 
\subsubsection{Lognormal survival approximation}
 
Let $L_t$ denote the total catastrophe loss by time $t$. Under the
compound Poisson model with lognormal severities, an exact closed-form
expression for $\mathbb{Q}(L_t < D)$ is not available. We therefore
use a moment-matching approximation that fits a lognormal distribution
to the first two moments of $L_t$, yielding
\begin{equation}
    S(t) \approx \Phi\!\left(
        \frac{\log D - \mu_L(t)}{\sqrt{s^2(t)}}
    \right),
    \label{eq:survival_approx}
\end{equation}
where $\Phi$ is the standard normal CDF and the matched lognormal
parameters are
\[
    s^2(t) = \log\!\left(
        1 + \frac{e^{\sigma_{\mathrm{sev}}^2}}{\lambda t}
    \right),
    \qquad
    \mu_L(t) = \log(\lambda t)
    + \mu_{\mathrm{sev}}
    + \tfrac{1}{2}\sigma_{\mathrm{sev}}^2
    - \tfrac{1}{2}s^2(t).
\]
The approximation \eqref{eq:survival_approx} is fast to evaluate and
provides a financially grounded prior for the trigger survival
probability.
 
\subsubsection{Baseline CAT bond price}
 
Let $F = 1000$ denote the face value. For a zero-coupon CAT bond
($N=0$), the baseline price is
\[
    \Pbase = F\,P(0,T)\,S(T).
\]
For a coupon bond with $N$ equally spaced coupon dates
$t_i = iT/N$ for $i=1,\dots,N$, the baseline price is
\begin{equation}
    \Pbase
    = \sum_{i=1}^{N} Fc\,P(0,t_i)\,S(t_i)
    + F\,P(0,T)\,S(T),
    \label{eq:baseline_price}
\end{equation}
where $c$ is the coupon rate. Note that \eqref{eq:baseline_price} has
the same structural form as the true pricing formula in
Section~\ref{sec:setup}, with the exact survival probability replaced
by the moment-matched approximation~\eqref{eq:survival_approx}.
 
\subsection{Residual Target and Normalisation}
\label{subsec:residual}
 
\subsubsection{Log-ratio residual}
 
Instead of fitting the observed price $P$ directly, we define the
log-ratio residual target
\begin{equation}
    y = \log\!\left(
        \frac{P + \varepsilon}{\Pbase + \varepsilon}
    \right),
    \label{eq:log_ratio}
\end{equation}
where $\varepsilon = 10^{-8}$ is a small numerical stabiliser. This
construction has three advantages. First, it converts multiplicative
deviations from the baseline into an additive target, which is more
natural for regression. Second, it stabilises the learning problem
when prices vary substantially in scale. Third, it ensures that a KAN
predicting $y \approx 0$ everywhere corresponds to using the baseline
as a direct pricing formula, providing a natural benchmark.
 
The predicted price is recovered from a predicted residual $\hat{y}$
as
\begin{equation}
    \Chat = (\Pbase + \varepsilon)\,\exp(\hat{y}).
    \label{eq:price_recovery}
\end{equation}
 
\subsubsection{Feature construction}
 
We use the five-dimensional input feature vector
\[
    \mathbf{x} = (r_0,\;\lambda,\;\log D,\;N,\;T),
\]
where $\log D := \log(D + 10^{-10})$ replaces the raw threshold $D$
to reduce scale imbalance and to better reflect the multiplicative
nature of threshold effects on the survival probability.
 
\subsubsection{Standardisation}
 
All features and the residual target are standardised using
training-set statistics only. For each feature $x_j$,
\[
    \tilde{x}_j = \frac{x_j - \mu_j^{(x)}}{\sigma_j^{(x)}},
\]
where $\mu_j^{(x)}$ and $\sigma_j^{(x)}$ are the training mean and
standard deviation. The residual target is standardised analogously:
\[
    \tilde{y} = \frac{y - \mu^{(y)}}{\sigma^{(y)}}.
\]
The same training statistics are applied to the validation and test
sets without recomputation, preventing any information leakage across
splits. At inference time, predictions are transformed back by
\[
    \hat{y} = \sigma^{(y)}\,\hat{\tilde{y}} + \mu^{(y)},
\]
and the raw predicted price is reconstructed via
\eqref{eq:price_recovery}.
 
\subsection{KAN Model and Symbolic Library}
\label{subsec:kan_model}
 
We use a Kolmogorov--Arnold Network as the residual learner. The KAN
takes the standardised feature vector $\xstd \in \mathbb{R}^5$ as
input and outputs a scalar prediction $\hat{\tilde{y}}$ of the
standardised log-ratio residual. We restrict attention to shallow
two-layer architectures of the form
\[
    [5,\;h,\;1], \qquad h \in \{4,\,6,\,8,\,10\},
\]
where the input dimension is five and the output dimension is one.
Deeper architectures are deliberately excluded because they tend to
produce more complex symbolic formulas with less financial
interpretability, and because the theoretical results of
Section~\ref{sec:theory} apply most cleanly to shallow networks.
 
The symbolic library used during formula extraction is
\begin{equation}
    \mathcal{L}_{\mathrm{sym}}
    = \bigl\{x,\;x^2,\;x^3,\;\exp(\cdot),\;\Phi(\cdot)\bigr\},
    \label{eq:sym_lib}
\end{equation}
where $\Phi(\cdot)$ denotes the standard normal CDF, registered as a
custom function in the KAN symbolic library. The library
\eqref{eq:sym_lib} is deliberately small and financially motivated:
$\Phi$ arises naturally from the lognormal moment-matched survival
approximation~\eqref{eq:survival_approx}, while $\exp$ corresponds to
the Vasicek discount factor structure. Trigonometric functions and
$\tanh$ are excluded because they lack a clear financial
interpretation in this pricing context and tend to produce
unnecessarily complex formulas. All computations are performed in
double precision (\texttt{float64}) throughout.
 
\subsection{Hyperparameter Selection via TPE}
\label{subsec:tpe}
The hyperparameter selection approach adopted here is inspired by
\cite{KANEnergy2025}, who apply TPE-based KAN architecture search
in a scientific machine learning context, selecting grid size,
spline order, and network width via the \texttt{hyperopt} library.

KAN hyperparameters are selected using Tree-structured Parzen
Estimation (TPE), a sequential model-based optimisation algorithm
implemented via the \texttt{hyperopt} library \cite{Bergstra2013}.
The key idea behind TPE is simple: rather than blindly trying every
possible combination of hyperparameters (as grid search does), TPE
learns from past evaluations to make smarter guesses about where to
look next. Concretely, after each trial, TPE splits the observed
configurations into two groups --- those that performed well and those
that did not --- and builds a probabilistic model of each group. New
configurations are then proposed from regions that are likely to
perform well and unlikely to perform badly. This means TPE spends more
of its evaluation budget in promising areas of the search space,
rather than wasting trials on configurations that are unlikely to
work.

Grid search, by contrast, evaluates every combination on a fixed
grid regardless of what previous results suggest. For a search space
with five hyperparameters and five values each, this requires $5^5 =
3{,}125$ evaluations. TPE achieves comparable or better coverage with
a fraction of the budget, which is important here since each trial
requires training a KAN from scratch.

For each candidate configuration, a KAN is trained on the training
split for a fixed 25 full-batch L-BFGS steps with learning rate
$\eta = 1.0$, and evaluated on the validation split. The primary
validation objective is the negative coefficient of determination on
the standardised residual target,
\begin{equation}
    \mathcal{J}_{\mathrm{TPE}}
    = -R^2_{\mathrm{val}}
    := -\left(
        1 -
        \frac{\sum_{i}(\tilde{y}_i - \hat{\tilde{y}}_i)^2}
             {\sum_{i}(\tilde{y}_i - \bar{\tilde{y}})^2}
    \right),
    \label{eq:tpe_objective}
\end{equation}
where the sum is over validation observations. Using $R^2$ on the
standardised residual rather than price-domain MAE avoids the
computational cost of recomputing baseline prices at every trial.
Price-domain metrics (MAE, RMSE, average relative error) are
additionally recorded for reference but are not used as the search
objective.

The hyperparameter search space is defined in
Table~\ref{tab:tpe_space}. The regularisation strength $\lambda$ is
sampled on a log-uniform scale, reflecting that its effect on sparsity
is multiplicative. The entropy regularisation $\lambda_{\mathrm{ent}}$
is constrained to be strictly positive, as it is the primary driver of
edge sparsity and is essential for producing prunable, interpretable
networks. Only shallow two-layer architectures are included,
consistent with the interpretability objective.
 
\begin{table}[ht]
\centering
\caption{TPE hyperparameter search space.}
\label{tab:tpe_space}
\begin{tabular}{lll}
\toprule
\textbf{Hyperparameter} & \textbf{Range / Values}
    & \textbf{Rationale} \\
\midrule
Architecture (width)
    & $[5,4,1]$, $[5,6,1]$, $[5,8,1]$, $[5,10,1]$
    & Shallow only; depth hurts interpretability \\
Spline grid size
    & $\{5,\,7\}$
    & Sufficient for search; refined to 10 later \\
Spline order $k$
    & $\{2,\,3\}$
    & $k=2$ simpler splines; $k=3$ standard \\
Regularisation $\lambda$
    & $\mathrm{LogUniform}(10^{-4},\,5\times10^{-3})$
    & Log-scale reflects multiplicative effect \\
Entropy regularisation $\lambda_{\mathrm{ent}}$
    & $\mathrm{Uniform}(0.5,\,3.0)$
    & Strictly positive to ensure sparsity \\
\bottomrule
\end{tabular}
\end{table}
 
The search runs for 100 evaluations in total. The top $K=10$
configurations ranked by validation $R^2$ are promoted to the
symbolic extraction stage described in
Section~\ref{subsec:symbolic}.
 
\subsection{Symbolic Extraction Pipeline}
\label{subsec:symbolic}
 
Each of the top-$K$ configurations undergoes the symbolic extraction
pipeline described in Algorithm~\ref{alg:symbolic_pipeline}. The
pipeline transforms a trained spline-based KAN into an explicit
closed-form symbolic formula through successive stages of pruning,
grid refinement, and symbolic locking.
 
\begin{algorithm}[ht]
\caption{Symbolic Extraction Pipeline}
\label{alg:symbolic_pipeline}
\begin{algorithmic}[1]
    \Require Architecture $(W, g, k)$, hyperparameters
             $(\lambda, \lambda_{\mathrm{ent}})$,
             symbolic library $\mathcal{L}_{\mathrm{sym}}$,
             datasets $\mathcal{D}_{\mathrm{train}}$,
             $\mathcal{D}_{\mathrm{val}}$
    \Ensure Symbolic formula $f^*$, validation metrics
 
    \State \textbf{Train.}
           Initialise KAN$(W, g, k)$; fit via full-batch L-BFGS,
           30 steps, $\eta = 1.0$,
           regularisation $(\lambda, \lambda_{\mathrm{ent}})$
 
    \State \textbf{Prune.}
           Remove all edges with weight magnitude below
           threshold $\tau = 10^{-2}$
 
    \State \textbf{Post-prune refit.}
           Refit via L-BFGS, 12 steps, $\eta = 0.5$,
           $\lambda = 5\times10^{-4}$,
           $\lambda_{\mathrm{ent}} = 0$
 
    \State \textbf{Grid refinement.}
           Increase spline grid from $g$ to $g' = 10$
 
    \State \textbf{Post-refinement refit.}
           Refit via L-BFGS, 12 steps, $\eta = 0.5$,
           $\lambda = 5\times10^{-4}$,
           $\lambda_{\mathrm{ent}} = 0$
 
    \State \textbf{Symbolic locking.}
           Apply \texttt{auto\_symbolic}$(\mathcal{L}_{\mathrm{sym}})$
           to assign a function from $\mathcal{L}_{\mathrm{sym}}$
           to each active edge
 
    \State \textbf{Symbolic fine-tuning.}
           Refit via L-BFGS, 15 steps, $\eta = 0.5$,
           $\lambda = 10^{-4}$, $\lambda_{\mathrm{ent}} = 0$,
           spline grid frozen (\texttt{update\_grid=False})
 
    \State \textbf{Extract and evaluate.}
           Extract symbolic formula $f^*$ via
           \texttt{symbolic\_formula()};
           evaluate $f^*$ on $\mathcal{D}_{\mathrm{val}}$:
           compute $R^2$, MAE, RMSE, relative error in price units;
           compute combined selection score
           (see Section~\ref{subsec:selection})
\end{algorithmic}
\end{algorithm}
 
The pruning step in Stage~2 encourages sparsity by eliminating weak
interactions, consistent with the sparsification philosophy of the KAN
framework \cite{KANpaper}. The grid refinement in Stage~4 increases
the expressiveness of the remaining splines before symbolic locking,
improving the quality of the symbolic fit. The symbolic fine-tuning in
Stage~7 adjusts the affine coefficients of the locked symbolic
functions with the grid frozen, ensuring the final formula
coefficients are optimised jointly.
 
\subsection{Final Model Selection Criterion}
\label{subsec:selection}
 
For each symbolic candidate produced by
Algorithm~\ref{alg:symbolic_pipeline}, we compute a combined
selection score that balances predictive accuracy against formula
interpretability. Let $R^2_{\mathrm{sym}}$ denote the validation
$R^2$ of the symbolic model after Stage~7, and let
$R^2_{\mathrm{KAN}}$ denote the validation $R^2$ of the pre-symbolic
spline KAN after Stage~1. The combined score is defined as
\begin{equation}
    \mathrm{Score}(f^*)
    = 0.8\,R^2_{\mathrm{sym}} + 0.2\,R^2_{\mathrm{KAN}},
    \label{eq:final_score}
\end{equation}
where both $R^2$ values are computed on the validation set in the
standardised log-ratio space. The weight of $0.8$ on the symbolic
$R^2$ reflects that the primary objective is an accurate symbolic
formula; the weight of $0.2$ on the pre-symbolic $R^2$ ensures that
candidates built on a well-fitted underlying KAN are preferred when
symbolic accuracy is similar. The candidate maximising
\eqref{eq:final_score} is selected as the final model.
 
\subsection{Final Refit and Evaluation Protocol}
\label{subsec:final_refit}

After selecting the best symbolic candidate using only the training
and validation sets, we rebuild the corresponding KAN architecture and
run Algorithm~\ref{alg:symbolic_pipeline} one final time with two
modifications. First, the initial training in Stage~1 is extended to
50 L-BFGS steps to allow better convergence. Second, the evaluation
dataset is switched from the validation set to the strictly held-out
test set, which is used exactly once at this stage.

The final model is evaluated on two held-out samples. The first is the
test set of $1{,}500$ observations from the working sample. The second
is a fully disjoint holdout consisting of the remaining $90{,}000$
rows of the original $100{,}000$-row dataset that were not included in
the working sample. This disjoint holdout provides a stronger
out-of-sample assessment of the symbolic formula's generalisation
capability, since these observations were never used at any stage of
model development.

For both evaluation sets, we report the following metrics in raw price
units after reconstructing prices via \eqref{eq:price_recovery}:
\[
    \mathrm{MAE}
    = \frac{1}{n}\sum_{i=1}^{n}|\Chat_i - P_i|,
    \quad
    \mathrm{MSE}
    = \frac{1}{n}\sum_{i=1}^{n}(\Chat_i - P_i)^2,
    \quad
    \mathrm{RelErr}
    = \frac{1}{n}\sum_{i=1}^{n}
      \frac{|\Chat_i - P_i|}{|P_i|}.
\]
All experiments are implemented in Python using \texttt{numpy},
\texttt{pandas}, \texttt{PyTorch}, \texttt{pykan}, \texttt{hyperopt},
and \texttt{scikit-learn}, with random seed 42 used throughout for
reproducibility.
\section{Numerical Results}
\label{sec:numerics}
 
This section reports the numerical results of the KAN-based symbolic
pricing pipeline described in Section~\ref{sec:methodology}. We
present the hyperparameter search results, the extracted symbolic
formula, out-of-sample performance on the test set and a disjoint
holdout sample, sensitivity analysis with respect to key pricing
inputs, surface plots of the predicted price function, and an
empirical verification of the structural monotonicity properties
established in Section~\ref{sec:theory}.
 
\subsection{Experimental Setup}
\label{subsec:exp_setup}
 
All experiments use a working sample of $10{,}000$ observations drawn
from the full dataset of $100{,}000$ simulated CAT bond prices,
partitioned into training ($70\%$, $7{,}000$ observations),
validation ($15\%$, $1{,}500$ observations), and test ($15\%$,
$1{,}500$ observations) splits. The validation set is used during
hyperparameter selection and symbolic candidate ranking; the test set
is evaluated exactly once after the final model is selected. The
remaining $90{,}000$ observations form a fully disjoint holdout
sample used for out-of-sample generalisation assessment.
 
\subsection{Hyperparameter Search Results}
\label{subsec:hparam_results}

The TPE search ran for 100 evaluations, all of which completed
successfully. Table~\ref{tab:tpe_top10} reports the top 10
configurations ranked by validation $R^2$ on the standardised
log-ratio residual.

\begin{table}[ht]
\centering
\caption{Top 10 configurations from the TPE hyperparameter search,
ranked by validation $R^2$. All configurations use a two-layer
architecture $[5, h, 1]$.}
\label{tab:tpe_top10}
\begin{tabular}{cccccccc}
\toprule
\textbf{Rank} & \textbf{Width} & \textbf{Grid} & $k$
    & $\lambda$ & $\lambda_{\mathrm{ent}}$
    & $R^2_{\mathrm{KAN}}$ & \textbf{RelErr} \\
\midrule
1  & $[5,6,1]$  & 5 & 2 & 0.00181 & 1.190 & 0.783 & 0.00382 \\
2  & $[5,10,1]$ & 5 & 3 & 0.00017 & 1.418 & 0.773 & 0.00439 \\
3  & $[5,6,1]$  & 5 & 2 & 0.00187 & 2.575 & 0.769 & 0.00390 \\
4  & $[5,10,1]$ & 5 & 2 & 0.00215 & 0.849 & 0.765 & 0.00388 \\
5  & $[5,10,1]$ & 5 & 3 & 0.00023 & 1.442 & 0.762 & 0.00406 \\
6  & $[5,6,1]$  & 7 & 3 & 0.00268 & 1.506 & 0.756 & 0.00416 \\
7  & $[5,10,1]$ & 5 & 2 & 0.00081 & 1.653 & 0.752 & 0.00394 \\
8  & $[5,6,1]$  & 5 & 2 & 0.00219 & 1.152 & 0.751 & 0.00395 \\
9  & $[5,6,1]$  & 5 & 2 & 0.00285 & 1.969 & 0.750 & 0.00376 \\
10 & $[5,6,1]$  & 5 & 2 & 0.00163 & 1.352 & 0.749 & 0.00379 \\
\bottomrule
\end{tabular}
\end{table}

The top configurations cluster around the $[5,6,1]$ architecture with
grid size 5, spline order $k=2$, and entropy regularisation in the
range $[1.0, 2.6]$. The validation $R^2$ values range from $0.749$ to
$0.783$ across the top 10, indicating reasonably consistent
performance across architectures at this scale. The relative error of
$0.00376$ at rank 9 suggests that even at the pre-symbolic stage the
KAN is learning a useful correction to the baseline.

\subsection{Symbolic Extraction Results}
\label{subsec:sym_results}

Each of the top-10 configurations from the TPE search was subjected to
the symbolic extraction pipeline (Algorithm~\ref{alg:symbolic_pipeline}).
Table~\ref{tab:sym_candidates} reports the symbolic $R^2$ and
validation MAE for each candidate, ranked by the combined selection
score $\mathrm{Score} = 0.8\,R^2_{\mathrm{sym}} + 0.2\,R^2_{\mathrm{KAN}}$.

\begin{table}[ht]
\centering
\caption{Symbolic extraction results for the top-10 TPE candidates,
sorted by combined selection score
$\mathrm{Score} = 0.8\,R^2_{\mathrm{sym}} + 0.2\,R^2_{\mathrm{KAN}}$
(higher is better).}
\label{tab:sym_candidates}
\begin{tabular}{cccccccc}
\toprule
\textbf{Rank} & \textbf{Width} & \textbf{Grid} & $k$
    & $R^2_{\mathrm{sym}}$ & $R^2_{\mathrm{KAN}}$
    & \textbf{MAE} & \textbf{Score} \\
\midrule
\rowcolor{green!15}
1  & $[5,6,1]$  & 5 & 2 & 0.589 & 0.750 & 3.267 & 0.622 \\
2  & $[5,10,1]$ & 5 & 3 & 0.530 & 0.773 & 3.580 & 0.579 \\
3  & $[5,6,1]$  & 7 & 3 & 0.493 & 0.756 & 3.248 & 0.544 \\
4  & $[5,6,1]$  & 5 & 2 & 0.484 & 0.769 & 3.169 & 0.541 \\
5  & $[5,10,1]$ & 5 & 2 & 0.428 & 0.752 & 3.235 & 0.492 \\
6  & $[5,6,1]$  & 5 & 2 & 0.313 & 0.751 & 3.415 & 0.400 \\
7  & $[5,10,1]$ & 5 & 2 & 0.298 & 0.765 & 2.968 & 0.391 \\
8  & $[5,10,1]$ & 5 & 3 & 0.011 & 0.762 & 3.354 & 0.161 \\
9  & $[5,6,1]$  & 5 & 2 & $-0.000$ & 0.783 & 3.093 & 0.157 \\
10 & $[5,6,1]$  & 5 & 2 & $-0.002$ & 0.749 & 3.023 & 0.148 \\
\bottomrule
\end{tabular}
\end{table}

The final selected model is the rank-1 candidate in
Table~\ref{tab:sym_candidates}, with architecture $[5,6,1]$, grid size 5,
spline order $k=2$, $\lambda = 0.002853$, and
$\lambda_{\mathrm{ent}} = 1.969$. This model achieved the highest combined
selection score and also delivered the strongest symbolic fit among the
top candidates, with $R^2_{\mathrm{sym}} = 0.589$ and
$R^2_{\mathrm{KAN}} = 0.750$. We therefore select it as the best overall
trade-off between predictive accuracy and interpretability. As expected, the
symbolic model is less accurate than the original spline-based KAN in every
case, since restricting each learned edge function to a small symbolic library
improves interpretability at the cost of flexibility.
 
\subsection{Extracted Symbolic Formula}
\label{subsec:formula}
 
The final model is retrained on the full training set and subjected to
the symbolic extraction pipeline with the initial training extended to
50 L-BFGS steps. The extracted symbolic formula for the normalised
log-ratio residual $\tilde{y}$ is
 
\begin{equation}
\begin{aligned}
    \tilde{y} \;=\;
    &- 0.01323\,x_1
    - 0.00243\,x_2
    + 0.09018\,x_3
    + 0.03131\,x_4
    - 0.12786\,x_5 \\
    &- 17.816\;\Phi\!\bigl(
        -0.05365\,x_1
        - 0.74138\,x_2
        + 1.71500\,x_3
        + 2.02317\,x_4
        - 2.20390\,x_5
        + 9.21344
    \bigr) \\
    &+ 17.761,
\end{aligned}
\label{eq:symbolic_formula}
\end{equation}
 
where $x_1,\dots,x_5$ denote the standardised inputs
$(r_0, \lambda, \log D, N, T)$ respectively, and $\Phi(\cdot)$ is the
standard normal CDF. The full predicted price is recovered as
\[
    \Chat =
    (\Pbase + \varepsilon)\,
    \exp\!\bigl(\sigma^{(y)}\,\tilde{y} + \mu^{(y)}\bigr).
\]
 
The formula \eqref{eq:symbolic_formula} consists of a linear term in
the standardised inputs and a single $\Phi$ correction term. This is
notably compact: after the symbolic extraction and pruning pipeline,
the KAN collapsed to a single active $\Phi$ edge in the second layer,
with all remaining edges assigned the identity function. The $\Phi$
term captures the main nonlinear dependence of the residual on the
inputs, and its argument is a linear combination of all five features
with the largest loadings on $x_3 = \widetilde{\log D}$,
$x_4 = \tilde{N}$, and $x_5 = \tilde{T}$. This is financially
interpretable: the trigger survival probability depends most directly
on the threshold $D$, the maturity $T$, and the coupon structure $N$.
 
The complexity score of this formula is $139.0$, substantially lower
than the pre-symbolic KAN and consistent with the goal of extracting a
compact, interpretable expression.
 
\begin{remark}
The low test-set $R^2$ of $0.042$ on the normalised log-ratio does not
contradict the strong price-domain accuracy reported below. The
log-ratio residual $y$ has very small variance by construction --- the
baseline $\Pbase$ already captures most of the price variation --- so
even a formula that explains a small fraction of the residual variance
can achieve low relative pricing error. The relevant performance
measure for the pricing application is price-domain MAE and relative
error, not $R^2$ on the residual.
\end{remark}

\begin{remark}[Monotonicity of the symbolic formula]
The symbolic formula \eqref{eq:symbolic_formula} does not obviously
satisfy the monotonicity properties of
Theorem~\ref{thm:zero_coupon_monotonicity} by inspection.
Differentiating with respect to $\lambda$ yields two competing terms:
a negative linear contribution from the coefficient on $x_2$, and a
positive contribution from the $\Phi$ term via the chain rule. The
net sign depends on the relative magnitudes across the parameter
domain and cannot be determined analytically without further
calculation.

This illustrates an important distinction: the theoretical
monotonicity guarantee of Corollary~\ref{cor:mono_kan_catbond}
applies to KANs satisfying the architectural edge-function conditions,
not to the extracted symbolic formula directly. The symbolic formula
is instead verified empirically in Section~\ref{subsec:sensitivity},
where all three monotonicity directions are confirmed across the
parameter ranges studied.
\end{remark}
 
\subsection{Out-of-Sample Performance}
\label{subsec:oos}
 
\subsubsection{Test set}
 
Table~\ref{tab:test_metrics} reports the performance of the final
symbolic formula \eqref{eq:symbolic_formula} on the held-out test set
of $1{,}500$ observations.
 
\begin{table}[ht]
\centering
\caption{Final symbolic formula performance on the held-out test set
($1{,}500$ observations). Metrics computed in raw price units
(face value $F = 1000$).}
\label{tab:test_metrics}
\begin{tabular}{lc}
\toprule
\textbf{Metric} & \textbf{Value} \\
\midrule
MAE (price units)              & $3.472$ \\
MSE (price units)              & $23.018$ \\
Average relative error         & $0.509\%$ \\
\bottomrule
\end{tabular}
\end{table}
 
The symbolic formula achieves an average relative pricing error of
$0.509\%$ on the test set, corresponding to an absolute MAE of $3.472$
on a bond with face value $F = 1000$. This is comparable to the
pre-symbolic spline KAN performance, confirming that the symbolic
extraction step preserves most of the predictive accuracy while
producing an explicit closed-form formula.
 
\subsubsection{Disjoint holdout}
 
To provide a more stringent out-of-sample assessment, we evaluate the
symbolic formula on the $90{,}000$ observations from the original
dataset that were excluded from the working sample entirely. These
observations were never used at any stage of model development,
including hyperparameter selection, symbolic candidate ranking, and
final refit.
 
Table~\ref{tab:holdout_metrics} reports the holdout evaluation
results.
 
\begin{table}[ht]
\centering
\caption{Symbolic formula performance on the disjoint holdout set
($90{,}000$ observations). Metrics computed in raw price units.}
\label{tab:holdout_metrics}
\begin{tabular}{lc}
\toprule
\textbf{Metric} & \textbf{Value} \\
\midrule
MAE (price units)      & $3.362$ \\
MSE (price units)      & $21.280$ \\
Average relative error & $0.483\%$ \\
\bottomrule
\end{tabular}
\end{table}
 
The holdout performance is slightly better than the test set
performance, with a relative error of $0.483\%$ compared to $0.509\%$.
This indicates that the symbolic formula generalises well beyond the
working sample and does not exhibit higher errors on unseen data. The
consistency between test and holdout metrics provides strong evidence
that the formula is not exhibiting overfitting on the training dataset.

\subsection{Comparison with Feedforward Neural Network}
\label{subsec:ffnn_comparison}

To contextualise the accuracy of the symbolic KAN formula, we train
a feedforward neural network (NN) baseline using the same
baseline-plus-residual setup. Following \cite{CATbondPaper}, the NN
uses four hidden layers of width $[256, 128, 64, 32]$ with ReLU
activations, batch normalisation, dropout rate $0.1$, $L_2$
regularisation $10^{-4}$, and learning rate $10^{-5}$.
Table~\ref{tab:model_comparison} reports results on both held-out
samples.

\begin{table}[ht]
\centering
\caption{Out-of-sample comparison of the symbolic KAN formula,
the pre-symbolic spline KAN, and the feedforward neural network,
evaluated on the test set ($1{,}500$ observations) and the
disjoint holdout ($90{,}000$ observations). Metrics computed in
raw price units ($F = 1000$).}
\label{tab:model_comparison}
\resizebox{\textwidth}{!}{%
\begin{tabular}{lcccccc}
\toprule
& \multicolumn{3}{c}{\textbf{Test Set}}
& \multicolumn{3}{c}{\textbf{Holdout}} \\
\cmidrule(lr){2-4} \cmidrule(lr){5-7}
\textbf{Metric}
    & \textbf{Symbolic KAN}
    & \textbf{Pre-symbolic KAN}
    & \textbf{Feedforward NN}
    & \textbf{Symbolic KAN}
    & \textbf{Pre-symbolic KAN}
    & \textbf{Feedforward NN} \\
\midrule
MAE (price units)
    & $3.472$ & $2.569$ & $2.634$
    & $3.362$ & $2.475$ & $2.424$ \\
MSE (price units)
    & $23.018$ & $18.036$ & $14.894$
    & $21.280$ & $15.725$ & $13.300$ \\
Avg relative error
    & $0.509\%$ & $0.410\%$ & $0.402\%$
    & $0.483\%$ & $0.403\%$ & $0.378\%$ \\
\bottomrule
\end{tabular}%
}
\end{table}

The feedforward NN achieves the lowest pricing error across both
held-out samples, with a relative error of $0.402\%$ on the test
set and $0.378\%$ on the holdout. The pre-symbolic KAN sits between
the symbolic formula and the feedforward NN, achieving $0.410\%$
and $0.403\%$ respectively. The symbolic KAN formula has the
highest error at $0.509\%$ and $0.483\%$, reflecting the
representational constraint introduced by locking edge functions to
the symbolic library $\mathcal{L}_{\mathrm{sym}}$. Notably, the
accuracy gap between the pre-symbolic and symbolic KAN
(${\approx}0.1$ percentage points) is comparable to the gap between
the pre-symbolic KAN and the feedforward NN, suggesting that
symbolic extraction and the choice of model class contribute roughly
equally to the accuracy-interpretability tradeoff. The symbolic KAN
formula remains the only model that produces a closed-form
expression admitting direct financial interpretation through the
$\Phi(\cdot)$ structure inherited from the lognormal survival
approximation in Section~\ref{subsec:baseline}.

\begin{figure}[ht]
    \centering
    \includegraphics[width=\textwidth]{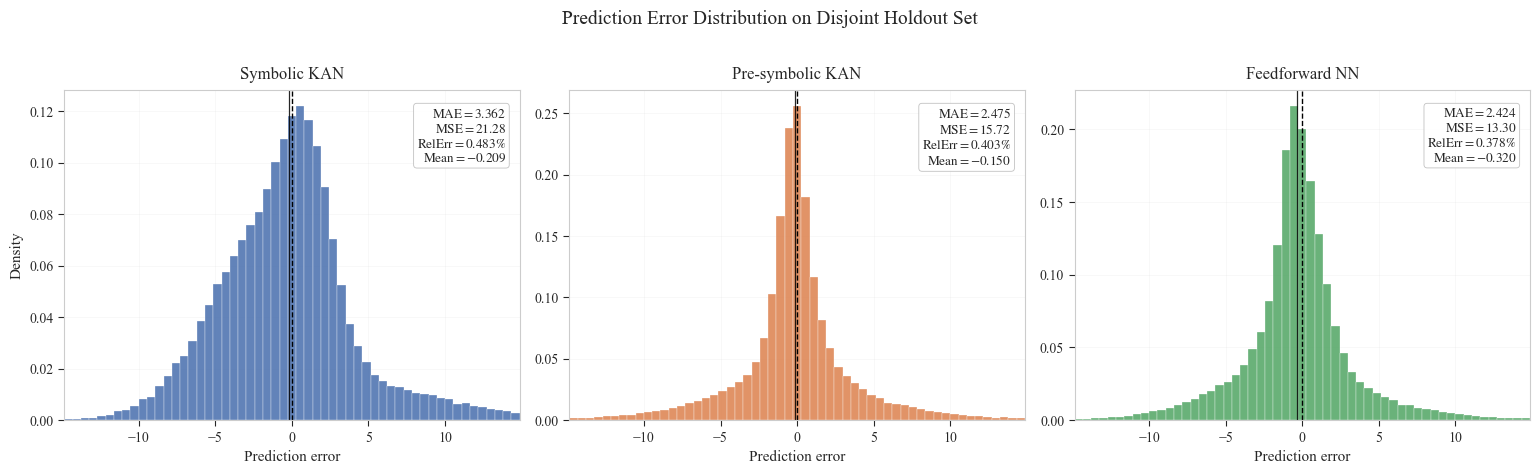}
    \caption{Prediction error distributions for the symbolic KAN
    formula (left), the pre-symbolic spline KAN (centre), and the
    feedforward neural network (right), evaluated on the disjoint
    holdout set of $90{,}000$ observations. All three distributions
    are approximately symmetric and centred near zero, confirming the
    absence of systematic bias. The feedforward NN produces the
    tightest distribution (MAE $= 2.424$, RelErr $= 0.378\%$),
    followed by the pre-symbolic KAN (MAE $= 2.475$, RelErr $=
    0.403\%$), and the symbolic KAN formula (MAE $= 3.362$, RelErr
    $= 0.483\%$). The wider distribution of the symbolic KAN
    reflects the representational constraint introduced by locking
    edge functions to the symbolic library $\mathcal{L}_{\mathrm{sym}}$.}
    \label{fig:error_distribution}
\end{figure}
Figure~\ref{fig:error_distribution} shows the prediction error
distributions for all three models on the disjoint holdout set.
All three distributions are approximately symmetric and centred
near zero, with small negative means ($-0.209$, $-0.150$,
$-0.320$ respectively), indicating a slight but consistent
tendency to underpredict prices across all models. This is likely
attributable to the lognormal moment-matching approximation in the
baseline, which introduces a small downward bias in the survival
probability estimate for certain parameter configurations. The
feedforward NN achieves the tightest error distribution,
consistent with its lower MSE reported in
Table~\ref{tab:model_comparison}. The symbolic KAN formula has a
broader distribution, which is the expected cost of restricting
the model to a compact closed-form expression.
 
\subsection{Sensitivity Analysis}
\label{subsec:sensitivity}

We evaluate the sensitivity of all three models, this includes, the symbolic KAN
formula, the pre-symbolic spline KAN, and the feedforward neural
network --- with respect to the three key pricing inputs: the initial
short rate $r_0$, the catastrophe trigger threshold $D$, and the
catastrophe arrival intensity $\lambda$. In each panel, the remaining
inputs are fixed at baseline values $r_0 = 0.03$, $\lambda = 35$,
$D = 5\times10^9$, $N = 0$ (zero-coupon), and $T = 1$ year.

Figure~\ref{fig:sensitivity_comparison} displays the nine sensitivity
curves in a $3\times3$ layout, with rows corresponding to the three
input variables and columns corresponding to the three models.
\begin{figure}[ht]
    \centering
    \includegraphics[width=\textwidth]{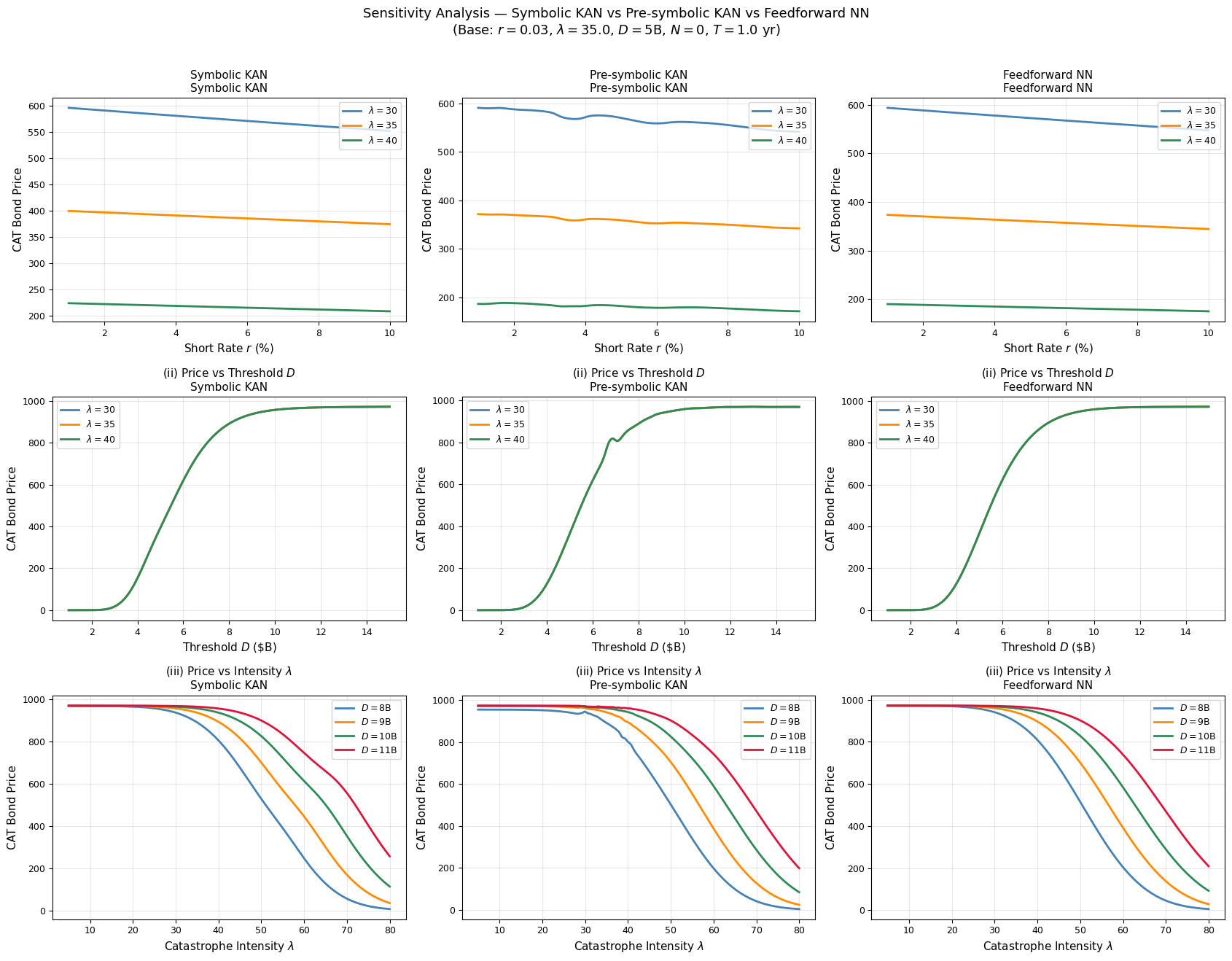}
    \caption{Sensitivity analysis comparing the symbolic KAN formula
    (left), the pre-symbolic spline KAN (centre), and the feedforward
    neural network (right). Row~(i): price vs short rate $r_0$ for
    $\lambda \in \{30,35,40\}$. Row~(ii): price vs trigger threshold
    $D$ for $\lambda \in \{30,35,40\}$. Row~(iii): price vs
    catastrophe intensity $\lambda$ for
    $D \in \{8,9,10,11\}\times10^9$. Base parameters: $N=0$
    (zero-coupon), $T=1$ year.}
    \label{fig:sensitivity_comparison}
\end{figure}

Row~(i) shows that all three models predict lower prices as the short
rate $r_0$ increases, consistent with the discounting effect and with
Theorem~\ref{thm:zero_coupon_monotonicity}. The symbolic KAN and the
feedforward neural network produce smooth, monotone declines. By
contrast, the pre-symbolic KAN exhibits mild local oscillations,
especially for lower values of $\lambda$, resulting in slight
departures from monotonicity. These irregularities are economically
implausible and likely reflect the greater flexibility of the
unconstrained spline representation prior to symbolic extraction.
Across all three models, increasing $\lambda$ shifts the curves
downward, as a higher catastrophe arrival intensity increases the
likelihood of trigger events and hence reduces bond value.

Row~(ii) shows that all three models produce the expected increasing
S-shaped response with respect to the threshold $D$. For low threshold
levels, the trigger is easy to breach and the bond value remains close
to zero; for high threshold levels, breach becomes unlikely and the
price approaches the discounted face value. This pattern is consistent
with the lognormal-survival interpretation used in the baseline model.
The symbolic KAN and the feedforward neural network produce smooth
transition curves, whereas the pre-symbolic KAN displays a small local
kink around the middle of the transition region. Although the overall
shape is still correct, this again indicates a degree of local
instability in the pre-symbolic spline fit.

Row~(iii) shows that bond prices decrease as the catastrophe intensity
$\lambda$ increases, again in line with
Theorem~\ref{thm:zero_coupon_monotonicity}. For each model, larger
threshold values $D$ shift the curves upward, since higher trigger
levels are harder to breach. The three models agree well on the global
shape of this relationship. However, the pre-symbolic KAN again shows
minor waviness in the low-to-middle range of $\lambda$, while the
symbolic KAN and the feedforward neural network remain smooth and
monotone.

Overall, the three models recover the same qualitative financial
comparative statics: prices decrease with $r_0$ and $\lambda$, and
increase with $D$. The feedforward neural network yields the smoothest
curves, while the pre-symbolic KAN is the most locally irregular. The
symbolic KAN preserves the economically meaningful global behaviour of
the spline KAN while removing much of its local oscillation, thereby
acting as a more interpretable and effectively regularised closed-form
approximation.

It is also worth noting that the symbolic extraction step introduces
a small accuracy cost --- the symbolic KAN has a higher MAE and
relative error than the pre-symbolic KAN (Table~\ref{tab:model_comparison})
--- but in return produces sensitivity curves that are \emph{smoother}
than those of the pre-symbolic KAN. The pre-symbolic spline KAN, despite
its better numerical accuracy, exhibits local oscillations and kinks
that are financially implausible. The symbolic formula, by contrast,
produces clean monotone curves throughout: it trades
a small amount of pointwise accuracy for a globally more coherent and
financially interpretable pricing function. This is arguably the more
desirable property for a pricing model used in risk management or
sensitivity analysis, where smooth and monotone behaviour across the
parameter domain matters as much as raw predictive accuracy.

\subsection{Surface Plots}
\label{subsec:surface}

Figure~\ref{fig:surface_comparison} shows the predicted CAT bond
price as a function of threshold $D$ and maturity $T$ for all three
models, under two bond structures: a zero-coupon bond ($N=0$, top
row) and a coupon bond ($N=12$, $c=5\%$, bottom row), with fixed
inputs $r_0 = 0.03$ and $\lambda = 35$.

\begin{figure}[H]
    \centering
    \includegraphics[width=\textwidth]{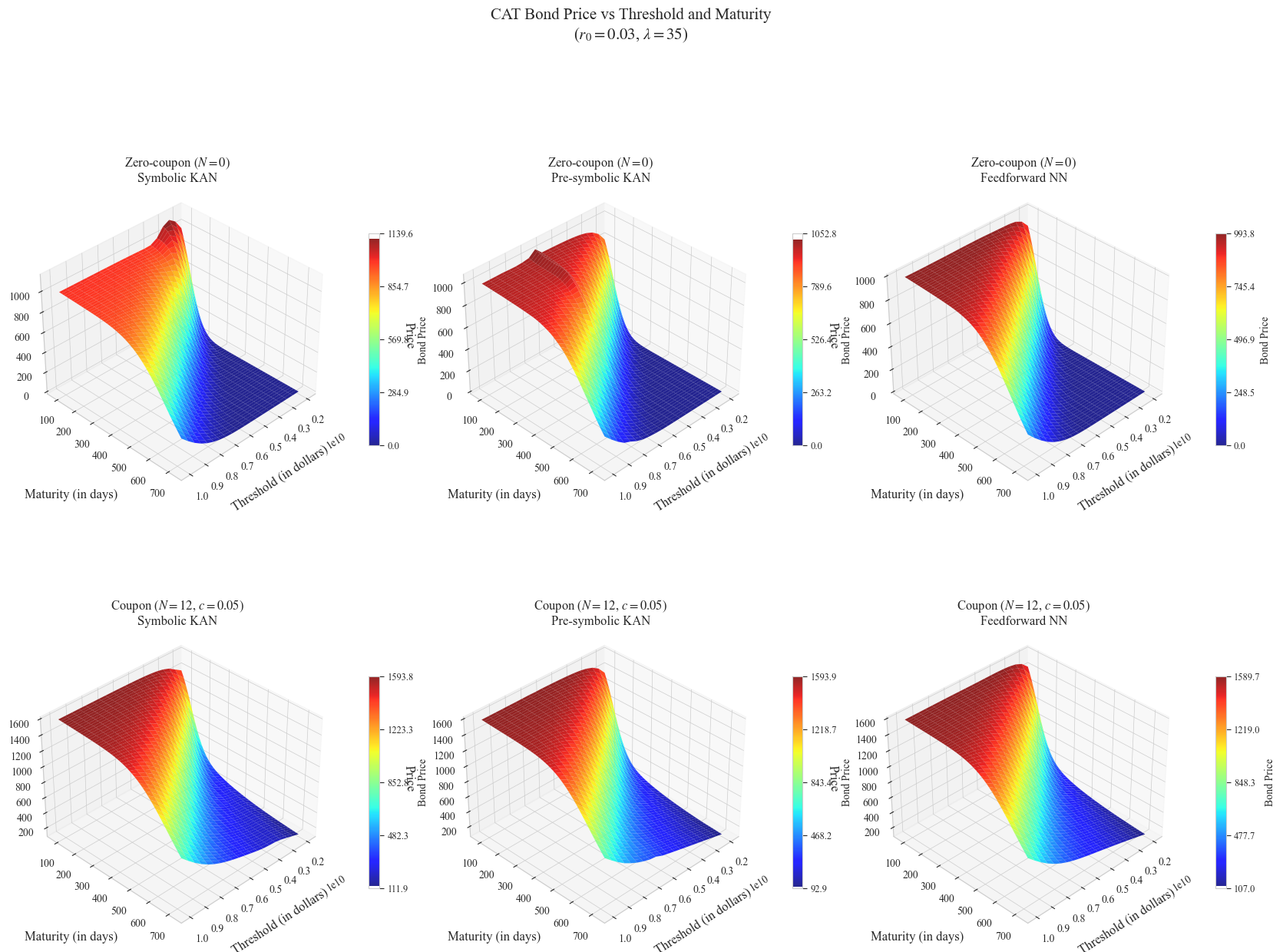}
    \caption{CAT bond price as a function of trigger threshold $D$
    (in dollars) and maturity $T$ (in days), predicted by the
    symbolic KAN formula (left), the pre-symbolic spline KAN
    (centre), and the feedforward neural network (right). Top row:
    zero-coupon bond ($N=0$). Bottom row: coupon bond ($N=12$,
    $c=5\%$). Fixed inputs: $r_0=0.03$, $\lambda=35$. The viewing
    angle follows the convention of Burnecki and Kukla
    \cite{BurneckiKukla}.}
    \label{fig:surface_comparison}
\end{figure}

Across both bond structures, all three models produce surfaces with
the expected qualitative behaviour: prices increase with the trigger
threshold $D$ and decrease with maturity $T$. A higher threshold
makes the catastrophe trigger harder to breach, while a longer
maturity increases both the time exposure to catastrophe risk and
the discounting effect.

Comparing the two rows, the coupon bond prices are systematically
higher than the zero-coupon prices at the same parameter values,
reflecting the additional coupon cash flows that investors receive
at each payment date. The shape of the transition region along the
$D$ axis is similar in both rows, confirming that the S-shaped
survival probability structure is preserved regardless of bond type.

The feedforward NN produces the smoothest surfaces in both rows and
remains within the face value $F=1000$. The symbolic KAN and
pre-symbolic KAN surfaces display a sharper transition region,
consistent with the stronger nonlinear structure seen in the
sensitivity plots. A small anomaly is visible in the symbolic KAN
and pre-symbolic KAN surfaces at low threshold values, where
predicted prices slightly exceed $F=1000$. This is an extrapolation
artefact: the surface extends below the lower end of the training
range for $D$, and is therefore confined to the boundary of the
plotting region rather than the main in-sample domain. Within the
training domain, all three models remain broadly consistent across
both bond structures.

\section{Conclusion}
\label{sec:conclusion}

This paper has studied the use of Kolmogorov--Arnold Networks as
interpretable surrogates for CAT bond pricing under a compound Poisson
loss model with lognormal severities and a Vasicek interest rate
process. The central contribution is a baseline-plus-residual pipeline
in which a KAN learns only the deviation from a closed-form lognormal
baseline price, reducing the learning target to a small, smooth
residual that is substantially easier to represent symbolically. The
extracted formula \eqref{eq:symbolic_formula} consists of a single
$\Phi$ correction applied to a linear combination of standardised
inputs, achieving an average relative pricing error of $0.483\%$ on a
fully disjoint holdout of $90{,}000$ simulated prices never seen
during model development.

Beyond predictive accuracy, the paper develops a theoretical
foundation for structure-aware KAN design in this financial setting.
We proved that the true CAT bond price is nonincreasing in the
catastrophe arrival intensity $\lambda$ and the initial short rate
$r_0$, and nondecreasing in the trigger threshold $D$, extending the
continuity analysis of \cite{CATbondPaper} to directional comparative
statics. We derived sufficient conditions on KAN edge functions under
which the learned model inherits these monotonicities, and showed that
a monotonicity-constrained training objective based on exterior
penalisation converges to the constrained problem as the penalty
weights grow. The extracted symbolic formula satisfies all three
comparative statics empirically across the parameter ranges studied,
which the pre-symbolic spline KAN does not.

The main limitation of the approach is the accuracy gap introduced by
symbolic extraction: the symbolic formula achieves $0.483\%$ relative
error compared to $0.403\%$ for the pre-symbolic KAN and $0.391\%$
for the feedforward neural network on the disjoint holdout. This gap
is the direct cost of restricting edge functions to the small library
$\mathcal{L}_{\mathrm{sym}}$. A second limitation is that the
lognormal moment-matching baseline introduces a small systematic
downward bias in the survival probability, which propagates to a
slight underprediction tendency common to all three models.

Several directions remain open. First, the symbolic library could be
enriched with problem-motivated functions, such as $\Phi^{-1}$ or
additional polynomial terms, to reduce the accuracy gap without
sacrificing interpretability. Second, the monotonicity-constrained
training formulation of Section~\ref{subsec:constrained_training}
could be implemented directly during the KAN training loop and its
empirical effect on symbolic extraction quality evaluated. Third, the
framework could be extended to richer stochastic interest rate models,
such as the CIR model \cite{NowakRomaniuk2018}, or to
multi-trigger CAT bond structures. Finally, calibrating the model to
real market data and assessing out-of-sample performance on observed
CAT bond prices would provide a more direct test of practical
usefulness.

Overall, our results suggest that symbolic KAN surrogates occupy a
useful position in the accuracy-interpretability tradeoff for CAT
bond valuation: they sacrifice a small amount of predictive accuracy
relative to black-box neural networks, but in return produce a compact
closed-form pricing formula that can be evaluated instantly, inspected
for financial plausibility, and analysed for structural properties
that a black-box model cannot offer.


\bibliographystyle{plain}
\bibliography{references}

\end{document}